\documentclass[10pt, conference, letterpaper]{IEEEtran}

\usepackage{url}
\usepackage{cite}

\usepackage{amsmath,amssymb,amsfonts} 
\usepackage{algorithmicx, algorithm}
\usepackage{algpseudocode}
\usepackage{graphicx} 
\usepackage{subcaption}
\usepackage{textcomp}
\usepackage{xcolor}

\usepackage{booktabs}
\usepackage{makecell} 

\usepackage{tikz}
\usetikzlibrary{arrows,decorations.markings, positioning}
\usepackage{tabu}
\usepackage{mathtools}
\usepackage[normalem]{ulem}
\usepackage[]{lineno}
\usepackage[]{scheduling}

\usepackage{listings} 

\newcommand{\setof}[1]{\left\{{#1}\right\}}

\makeatletter
\DeclareRobustCommand{\cev}[1]{%
	{\mathpalette\do@cev{#1}}%
}
\newcommand{\do@cev}[2]{%
	\vbox{\offinterlineskip
		\sbox\z@{$\m@th#1 x$}%
		\ialign{##\cr
			\hidewidth\reflectbox{$\m@th#1\vec{}\mkern4mu$}\hidewidth\cr
			\noalign{\kern-\ht\z@}
			$\m@th#1#2$\cr
		}%
	}%
}

\algnewcommand\algorithmicforeach{\textbf{for each}}
\algdef{S}[FOR]{ForEach}[1]{\algorithmicforeach\ #1\ \algorithmicdo}

\makeatother

\newcommand{\Tbb}{\mathbb{T}}
\newcommand{\Rbb}{\mathbb{R}}

\newcommand{\Nbb}{\mathbb{N}}
\newcommand{\Scal}{\mathcal{S}}
\newcommand{\Jcal}{\mathcal{J}}

\usepackage{amsthm}

\newtheorem{thm}{Theorem}

\newtheorem{lemma}[thm]{Lemma} 

\newtheorem{obs}[thm]{Observation}

\theoremstyle{definition}
\newtheorem{defn}[thm]{Definition}

\theoremstyle{remark}

\newtheorem{exmpl}[thm]{Example}

\newcommand{\Jbb}{\mathbb{J}}
\newcommand{\Cbb}{\mathbb{C}}
\newcommand{\mJbb}{{m_{\Jbb}}}
\newcommand{\mTbb}{{m_{\Tbb}}}

\newcommand{\sbar}{\bar{s}}
\newcommand{\fbar}{\bar{f}}
\newcommand{\Cbar}{\bar{C}}
\newcommand{\Sbar}{\bar{S}}
\newcommand{\Scalbar}{\bar{\Scal}}

\newcommand{\Wbar}{\bar{W}}

\newcommand{\ex}{\mathit{ex}}
\newcommand{\exbar}{\bar{\mathit{ex}}}

\newcommand{\Ex}{\mathit{Ex}}
\newcommand{\Rel}{\mathit{Rel}}

\newcommand{\lpi}{<_\pi}

\newcommand{\gpi}{>_\pi}

\newcommand{\rbar}{\bar{r}}
\newcommand{\enf}{\mathit{enf}}
\newcommand{\Seg}{\mathit{Seg}}

\usepackage{colonequals} 

\title{Eliminating Timing Anomalies in Scheduling Periodic Segmented Self-Suspending Tasks with Release Jitter}

\author{
	\IEEEauthorblockN{%
		Ching-Chi Lin\IEEEauthorrefmark{1},
		Mario G\"{u}nzel\IEEEauthorrefmark{1},
		Junjie Shi\IEEEauthorrefmark{1},
		Tristan Taylan Seidl\IEEEauthorrefmark{1},
		Kuan-Hsun Chen\IEEEauthorrefmark{3},
		and Jian-Jia Chen\IEEEauthorrefmark{1}\IEEEauthorrefmark{2}
	}
	\IEEEauthorblockA{%
		\IEEEauthorrefmark{1} Technical University of Dortmund, Dortmund, Germany\\
		\IEEEauthorrefmark{2} Lamarr Institute for Machine Learning and Artificial Intelligence, Dortmund, Germany\\
		Email: \{chingchi.lin, mario.guenzel, junjie.shi, tristan.seidl, jian-jia.chen\}@tu-dortmund.de\\
		\IEEEauthorrefmark{3} University of Twente, Enschede, Netherlands\\
		Email: k.h.chen@utwente.nl
	}
}

\begin{document}

\maketitle

\begin{abstract}
	Ensuring timing guarantees for every individual tasks is critical in real-time systems.
	Even for periodic tasks, providing timing guarantees for tasks with segmented self-suspending behavior is challenging due to timing anomalies, i.e., the reduction of execution or suspension time of some jobs increases the response time of another job.
	The release jitter of tasks can add further complexity to the situation, affecting the predictability and timing guarantees of real-time systems.
	The existing worst-case response time analyses for sporadic self-suspending tasks are only over-approximations and lead to overly pessimistic results.
	In this work, we address timing anomalies without compromising the worst-case response time (WCRT) analysis when scheduling periodic segmented self-suspending tasks with release jitter.
	We propose two treatments: \emph{segment release time enforcement} and \emph{segment priority modification}, and prove their effectiveness in eliminating timing anomalies. 
	Our evaluation demonstrates that the proposed treatments achieve higher acceptance ratios in terms of schedulability compared to state-of-the-art scheduling algorithms.
	Additionally, we implement the segment-level fixed-priority scheduling mechanism on RTEMS and verify the validity of our \emph{segment priority modification} treatment.
	This work expands our previous conference publication at the 29th IEEE Real-Time and Embedded Technology and Applications Symposium (RTAS'23), which considers only periodic segmented self-suspending tasks without release jitter.
\end{abstract}

\section{Introduction}\label{sec:intro}

Due to the prevalence of self-suspensions in real-time and cyber-physical systems, scheduling tasks with self-suspension has become a significant research topic since its inception in 1988~\cite{DBLP:conf/rtss/RajkumarSL88}.
A self-suspending task may temporarily yield the processor to other tasks for a certain period of time, in addition to being preempted by higher-priority tasks. 
Self-suspension can arise due to various factors, such as I/O- or memory-intensive tasks~\cite{Kang:rtss07,Kato_2011}, multiprocessor synchronization~\cite{DBLP:books/sp/22/Brandenburg22}, hardware acceleration by using coprocessors and computation offloading~\cite{nimmagadda2010real,DBLP:conf/ecrts/TomaC13,DBLP:conf/dac/LiuCTKD14}, scheduling of parallel tasks~\cite{DBLP:conf/sies/FonsecaNNP16,DBLP:conf/ecrts/UeterGBC21}, real-time tasks in multi-core systems with shared memory~\cite{DBLP:conf/dac/HuangCR16}, timing analysis of deferrable servers~\cite{LiuChen:rtss2014,DBLP:conf/ecrts/ChenGJBC22}, dynamic reconfigurable FPGAs for real-time applications~\cite{DBLP:conf/rtss/BiondiBPRMB16}, and real-time communication for networks-on-chip~\cite{DBLP:conf/rtcsa/UeterCBVM20}. 
The suspension time between consecutive computation segments can vary widely depending on the application, ranging from microseconds to hundreds of milliseconds or even seconds.

Ensuring timing guarantees for individual tasks is crucial in real-time systems.
There are two primary approaches for estimating these guarantees.
The first approach involves performing a schedulability test tailored to the specific scheduling algorithm to determine if a given task set can be successfully scheduled.
Alternatively, the \textit{worst-case response time} (WCRT) of each task can be analyzed to validate potential deadline violations. 
However, the validation of timing guarantees becomes challenging when tasks exhibit \emph{self-suspending} behavior.
As pointed out in the review paper by Chen~et~al.~\cite{suspension-review-jj}, ``\emph{allowing tasks to self-suspend ... conversely has the effect that key insights underpinning the analysis of non-self-suspending tasks no longer hold}.''.
Consequently, the unintuitive timing behavior of self-suspending tasks has led to several flaws in the existing literature~\cite{suspension-review-jj,GuenzelC20_RTSJournal,GuenzelC21_RTSJournal}.

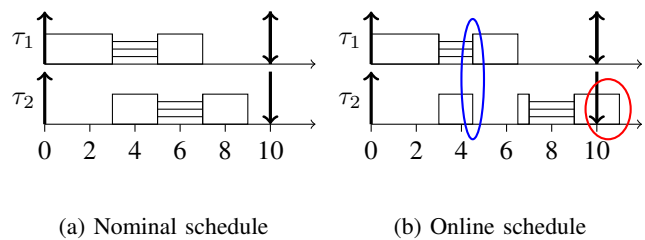
\begin{figure}[tb]
	\centering	
	\begin{subfigure}{.49\linewidth}
	\centering
	\begin{tikzpicture}[yscale=0.4, xscale=0.3]
	\begin{scope}[shift={(0,2)}] 
		\taskname{$\tau_1$}
		
		\timeline{0}{12}{}
		
		\releases{0,10}
		\deadlines{10}
		
		\exec{0}{3}
		\exec{5}{7}		
		
		\susp{3}{5}		
	\end{scope}
	
	\begin{scope}[shift={(0,0)}] 
		\taskname{$\tau_2$}
		
		\timeline{0}{12}{}
		\labelling{0}{10}{2}{0}
		
		\releases{0}
		\deadlines{10}
		
		\exec{3}{5}
		\exec{7}{9}		
		
		\susp{5}{7}
	\end{scope}
	\end{tikzpicture}
	\label{fig:timing_anomaly_a}
	\caption{Nominal schedule}
	\end{subfigure}%
	\begin{subfigure}{.49\linewidth}
		\centering
		\begin{tikzpicture}[yscale=0.4, xscale=0.3]
		\begin{scope}[shift={(0,2)}] 
			\taskname{$\tau_1$}
			
			\timeline{0}{12}{}
			
			\releases{0,10}
			\deadlines{10}
			
			\exec{0}{3}
			\exec{4.5}{6.5}		
			
			\susp{3}{4.5}		
		\end{scope}
		
		\begin{scope}[shift={(0,0)}] 
			\taskname{$\tau_2$}
			
			\timeline{0}{12}{}
			\labelling{0}{10}{2}{0}
			
			\releases{0}
			\deadlines{10}
			
			\exec{3}{4.5}
			\exec{6.5}{7}
			\exec{9}{11}		
			
			\susp{7}{9}
			\draw[blue, thick] (4.5,1.5) ellipse (0.5cm and 2cm);
			\draw[red, thick] (10.5,0.5) ellipse (1cm and 1cm);
		\end{scope}
		\end{tikzpicture}
		\label{fig:timing_anomaly_b}
		\caption{Online schedule}
		\end{subfigure}
	\caption{Example of a timing anomaly.  Assume that segments from $\tau_1$ have higher priorities than segments from $\tau_2$.  (a) A nominal schedule generated based on the WCET and the maximum suspension time of the segments;  (b) Task $\tau_2$ misses its deadline due to the suspension interval from $\tau_1$ finishes earlier at time \(4.5\) instead of $5$.}
	\label{fig:timing_anomaly}	
\end{figure}

Unlike tasks without self-suspending behaviors, where the WCRT is achieved when all jobs execute for their \textit{worst-case execution time} (WCET), self-suspending behavior can introduce \emph{timing anomalies}. 
A timing anomaly occurs when decreasing the execution or suspension time for certain jobs results in an increase in the response time for another job. 
Since the actual execution or suspension time of a segment can be less than its WCET or maximum suspension time, its succeeding segments may become ready for execution earlier and interfere with segments from other tasks.
An example of such timing anomalies is illustrated in Figure~\ref{fig:timing_anomaly}, where a shorter suspension of a higher-priority task results in increased interference with a lower-priority task. 
Counterintuitively, enforcing the suspension time to the maximum can improve the schedulability~\cite{DBLP:conf/rtcsa/SchonbergerHBCC18}.

The release jitter of tasks can be another source of timing anomalies.
For periodic tasks with release jitter, the actual release time of a job may deviate from its expected release time, leading to irregularities in the execution patterns.
The release jitter can be caused by various factors, such as external events, shared resources, interrupt handling, or variations in system workload.

To estimate timing guarantees for tasks with release jitter, a common approach involves considering the \emph{maximum release jitter}, i.e., jobs are released at their latest possible release time, while determining if a task set is schedulable.
However, this approach is pessimistic and rarely occurs in practice.
Nevertheless, a shorter actual release jitter may result in a timing anomaly, as shown in Figure~\ref{fig:release_jitter_anomaly}.

\begin{figure}[tb]
	\centering	
	\begin{subfigure}{.5\linewidth}
	\centering
	\begin{tikzpicture}[yscale=0.4, xscale=0.3]
	\begin{scope}[shift={(0,2)}] 
		\taskname{$\tau_1$}
		
		\timeline{0}{12}{}
		
		\releases{0,10}
		\deadlines{10}
		
		\exec{2}{3}
		\exec{5}{8}		
		
		\susp{3}{5}

		\draw[->, dashed, very thick, color=blue] (2,0) -- (2, 1.75);
	\end{scope}
	
	\begin{scope}[shift={(0,0)}] 
		\taskname{$\tau_2$}
		
		\timeline{0}{12}{}
		\labelling{0}{10}{2}{0}
		
		\releases{0}
		\deadlines{10}
		
		\exec{1}{2}
		\exec{3}{5}
		\exec{8}{9.6}		
		
		\susp{5}{8}

		\draw[->, dashed, very thick, color=blue] (1,0) -- (1, 1.75);
	\end{scope}
	\end{tikzpicture}
	\label{fig:jitter_a}
	\caption{Nominal schedule}
	\end{subfigure}%
	\begin{subfigure}{.5\linewidth}
		\centering
		\begin{tikzpicture}[yscale=0.4, xscale=0.3]
		\begin{scope}[shift={(0,2)}] 
			\taskname{$\tau_1$}
			
			\timeline{0}{12}{}
			
			\releases{0,10}
			\deadlines{10}
			
			\exec{1}{2}
			\exec{4}{7}		
			
			\susp{2}{4}
			\draw[->, dashed, very thick, color=blue] (1,0) -- (1, 1.75);
		\end{scope}
		
		\begin{scope}[shift={(0,0)}] 
			\taskname{$\tau_2$}
			
			\timeline{0}{12}{}
			\labelling{0}{10}{2}{0}
			
			\releases{0}
			\deadlines{10}
			
			\exec{2}{4}
			\exec{7}{8}
			\exec{11}{12}
			
			\susp{8}{11}
			\draw[->, dashed, very thick, color=blue] (1,0) -- (1, 1.75);
			\draw[red, thick] (10.5,0.5) ellipse (1cm and 1cm);
		\end{scope}
		\end{tikzpicture}
		\label{fig:jitter_b}
		\caption{With release jitter}
		\end{subfigure}
	\caption{Timing anomaly caused by release jitter.  Assume that segments from $\tau_1$ have higher priorities than segments from $\tau_2$.  (a) A feasible nominal schedule considering the maximum release jitter of both tasks;  (b) Task $\tau_2$ misses its deadline due to a smaller actual release jitter from $\tau_1$.}
	\label{fig:release_jitter_anomaly}	
\end{figure}
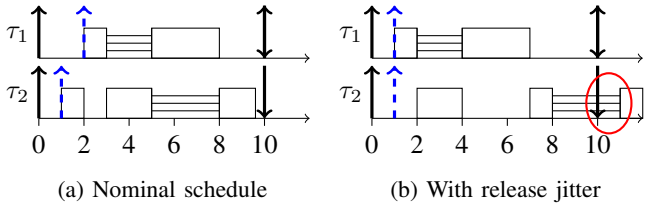

In this work, we focus on preemptive scheduling of segmented self-suspension periodic tasks on a uniprocessor system, which was previously analyzed by over-approximation with sporadic task model~\cite{DBLP:conf/rtcsa/SchonbergerHBCC18,RTCSA-BletsasA05, RTSS-ChenL14,DBLP:conf/rtcsa/PengF16,WC16-suspend-DATE, Kim2016,ecrts15nelissen,Bruggen16RTNS,DBLP:conf/ecrts/ChenHHMB19}.
Despite the prevailing research results of self-suspending tasks, most of them focus on the \textit{sporadic} task model, where the jobs of a task are specified with a minimum inter-arrival time.
Nevertheless, periodic task activation is a prevalent industry practice, with $82\%$ of the investigated systems adhering to this approach in an empirical study~\cite{DBLP:journals/rts/AkessonNNAD22}. 
To the best of our knowledge, there are only two works considering periodic tasks with self-suspension: Yalcinkaya et al.~\cite{DBLP:conf/date/YalcinkayaNB19} for segmented self-suspending tasks scheduled non-preemptively; and G\"unzel et~al.~\cite{guenzel2020sched_test_edf} for dynamic self-suspending tasks scheduled preemptively.
These works indicate that a dedicated analysis of the periodic job release pattern of a task can be beneficial. 

\textbf{Contributions}: 
To ensure the schedulability of segmented self-suspension periodic tasks with release jitter, our proposed solutions are based on the following two steps:
\begin{itemize}
	\item \textbf{Step 1}: A \emph{nominal schedule} is constructed and recorded offline based on the maximum release jitter of tasks, the worst-case execution time of the computation segments, and the maximum suspension time of the suspension intervals.
	\item \textbf{Step 2}: The \emph{online schedule} refers to the nominal schedule to make scheduling decisions \emph{without any risk of timing anomalies} so that the schedulability (feasibility) of the online schedule is guaranteed as long as schedulability (feasibility) of the nominal schedule is guaranteed.
\end{itemize}

In our previously published version at the 29th IEEE Real-Time and Embedded Technology and Applications Symposium (RTAS'23)~\cite{DBLP:conf/rtas/LinGSSCC23}, we considered only periodic segmented self-suspending tasks \textbf{without} release jitter, and made the following contributions:

\begin{itemize}
	\item We proposed two treatments, \emph{segment  release time enforcement} and \emph{segment priority modification}, regarding anomaly-free online scheduling.
	With \emph{segment release time enforcement}, a task segment is enforced to start its execution \textbf{no earlier} than its release time in the nominal schedule.
	For \emph{segment priority modification}, the priorities of the segments are adjusted based on their nominal finishing times, so that a segment with an earlier nominal finishing time cannot be interfered by segments with later nominal finishing times.
	We proved that the schedulability of the segmented self-suspension periodic tasks \textbf{without} release jitters under these two treatments is guaranteed if and only if the nominal schedule ensures that all jobs meet their deadlines, resulting in an exact schedulability test without any over-approximation.	
	\item The empirical results demonstrated that the proposed treatments achieve higher acceptance ratios in terms of schedulability compared to state-of-the-art over-approximations under different task set configurations.
	\item We implemented the segment-level fixed-priority scheduling mechanism on RTEMS~\cite{rtems}, an open-source Real-Time Operating System (RTOS), so that the treatment \emph{segment priority modification} can be performed on RTEMS.
	We discuss the implementation details and showcase the validity of the treatment with an example.
\end{itemize}

By incorporating segmented self-suspension periodic tasks \textbf{with} release jitter, we expand the scope of our prior work and make the following supplementary contributions in this journal version:
\begin{itemize}
	\item 
	We distinguish the release time of a job from its first computation segment, and model the difference, i.e., the release jitter, as a suspension interval before the first computation segment in Section~\ref{sec:task_model_jitter}.
	We modify the two proposed treatments and prove their effectiveness in eliminating timing anomalies for tasks with release jitter.
	\item In Section~\ref{sec:eval_jitter}, we assess the influence of release jitter on the acceptance ratio of nominal schedules.
	Our findings indicate that incorporating the maximum release jitter into consideration during the generation of nominal schedules leads to only a marginal reduction in the acceptance ratio for high-utilization task sets, compared to schedules generated without considering release jitter. 
	This suggests that the proposed treatments effectively mitigate timing anomalies for tasks with release jitter without substantially compromising schedulability.		
\end{itemize}

\section{System Model and Background}\label{sec:sys_model}

In this paper, we focus on preemptive schedules of \textit{segmented self-suspension periodic tasks with bounded release jitter} on a uniprocessor system.
First, we introduce the task model and define the notations for segmented self-suspending task \textbf{without} release jitter in Section~\ref{sec:task_model}.
In Section~\ref{sec:task_model_jitter}, we extend the task model to consider release jitter.
Definitions and observations on segment-level fixed-priority preemptive scheduling are presented in Section~\ref{sec:scheduling}.
Table~\ref{tb:notations} summarizes the notations being used in this paper.

\begin{table}[tb]
	\caption{Notations}
	\label{tb:notations}
	\resizebox{\linewidth}{!}{
	\begin{tabular}{l|l}
		\toprule
		Notation			& Description	\\ \midrule
		$\tau$				& a segmented self-suspending task		\\
		$C_{\tau}^i$		& the WCET of the $i$-th segment in task $\tau$.	\\
		$S_{\tau}^i$		& \makecell[l]{the maximum suspension time of the $i$-th \\ suspending interval in task $\tau$}	\\
		$\Jcal{\tau}$		& the maximum release jitter of jobs released by task \(\tau\) \\
		\midrule
		$J$					& a job released by a task	\\
		$\gamma$			& a computation segment within a job	\\		
		$C_{\gamma}$, $\Cbar_{\gamma}$	& the WCET / actual execution time of segment $\gamma$.	\\
		$S_{\gamma}$, $\Sbar_{\gamma}$	& \makecell[l]{the maximum / actual suspension time after \\ segment $\gamma$}	\\		
		$s_{J}$,$\sbar_{J}$, $s_{\gamma}$, $\sbar_{\gamma}$	& \makecell[l]{the nominal / actual starting time of job $J$ \\ / segment $\gamma$.}	\\
		$f_{J}$,$\fbar_{J}$, $f_{\gamma}$, $\fbar_{\gamma}$	& \makecell[l]{the nominal / actual finishing time of job $J$ \\ / segment $\gamma$.}	\\
		$r_{\gamma}, \rbar_{\gamma}, r^\enf_{\gamma}$ & the nominal / actual / enforced release time of $\gamma$ \\
		$\ex(\gamma), \exbar(\gamma)$		& \makecell[l]{the set of time points at which $\gamma$ is executed \\in the nominal / actual schedule}\\
		$W_{\gamma}(r, t), \Wbar_{\gamma}(r, t)\!\!\!$ & \makecell[l]{the total amount of time segments with a higher \\ priority than $\gamma$ are executed during $[r,t)$ in $\Scal$ / $\Scalbar$.} \\
		\midrule
		$\mJbb$, $\mTbb$	& the mapping from segment to job and job to task		\\
		$\Scal, \Scalbar$	& the nominal schedule and the actual schedule	\\
		$\gpi, >_P$			& total priority / preference ordering \\
		$\mu()$				& the Lebesgue measure \\
	\bottomrule
	\end{tabular}
	}
\end{table}

\subsection{Self-Suspending Task}\label{sec:task_model}

Two self-suspension models are mostly studied: the \emph{segmented} self-suspension model~\cite{GuenzelC21_RTSJournal,DBLP:conf/rtcsa/SchonbergerHBCC18,RTCSA-BletsasA05, RTSS-ChenL14,DBLP:conf/rtcsa/PengF16,WC16-suspend-DATE, Kim2016,ecrts15nelissen,Bruggen16RTNS,DBLP:conf/ecrts/ChenHHMB19,DBLP:conf/date/YalcinkayaNB19} which assumes a fixed iterating pattern of execution segments and suspension intervals, and the \emph{dynamic} self-suspension model~\cite{LiuChen:rtss2014,GuenzelC20_RTSJournal,guenzel2020sched_test_edf,ECRTS-AudsleyB04,huangpass:dac2015, ChenECRTS2016-suspension,DBLP:conf/ecrts/Devi03,DBLP:conf/rtss/GunzelUC21,DBLP:conf/ecrts/AromoloBN22} that allows arbitrary execution and suspension interleaves, meaning that a job can self-suspend as long as its total suspension time does not exceed its maximum suspension time. 

The dynamic self-suspension model can cater to any task model with self-suspending behavior, but its flexibility
results in a very pessimistic analysis if the suspension behaviors of the tasks can be described more precisely using the segmented self-suspension model.  Detailed discussions of self-suspension can be found in the survey papers by Chen~et~al.~\cite{suspension-review-jj,DBLP:conf/rtcsa/ChenBH017}.

We adopt the \textit{segmented self-suspension} task model in this work.
The system consists of several (finitely many) tasks $\Tbb$.
Each task $\tau \in \Tbb$ releases jobs successively, and each job $J$ is divided into several computation segments.
We denote by $\Jbb$ the set of all jobs and by $\Cbb$ the set of all computation segments.
Each computation segment belongs to one job and each job belongs to one task, therefore we have the following mappings:
\begin{equation}
	\Cbb \overset{\mJbb}{\longrightarrow} \Jbb \overset{\mTbb}{\longrightarrow} \Tbb, \nonumber
\end{equation}
where $\mJbb$ maps a segment $\gamma \in \Cbb$ to its corresponding job $\mJbb(\gamma) \in \Jbb$, and the function $\mTbb$ maps a job $J$ to its releasing task $\mTbb(J) \in \Tbb$.

Each \textbf{segmented self-suspending task} $\tau$ consists of $M_{\tau}$ computation segments and $M_{\tau}-1$ suspension intervals, $M_{\tau} \geq 1$.
We denote $\tau$ as 
\begin{equation}
	\tau = (\Ex_{\tau}, \Rel_\tau), \nonumber
\end{equation}
where $\Ex_\tau$ describes the execution behavior of $\tau$ and $\Rel_\tau$ describes the release behavior of $\tau$.
To be more specific, $\Ex_\tau = (C_{\tau}^0, S_{\tau}^0, C_{\tau}^1, S_{\tau}^1, \ldots, S_{\tau}^{M_i-2}, C_{\tau}^{M_i-1})$, where $C_{\tau}^j >0 $ is the WCET of the $j$-th computation segment in $\tau$, $S_{\tau}^j > 0$ is the maximum suspension time of the $j$-th suspension interval in $\tau$.
Moreover, $\Rel_\tau = (\rho_\tau^1, \rho_\tau^2, \dots) \in \Rbb^\Nbb$ is the list of all release times ordered in increasing order.
For example, a periodic task $\tau$ with a period of $10$ and the first release at time $0$ has $\Rel_\tau = (0, 10, 20, \dots)$.

We assume that the segments are scheduled according to a \emph{Segment-level Fixed-Priority} (S-FP) scheduling mechanism, which indicates that there is a total priority ordering $\lpi$ of the segments $\Cbb$.
If a segment $\gamma \in \Cbb$ has a higher priority than another segment $\omega \in \Cbb$, then we write $\omega \lpi \gamma$.  
We distinguish two different schedules,
$\Scal$ is the \emph{nominal} schedule that is obtained when each segment executes its worst-case execution time (WCET) and is suspended for its maximum suspension time.
$\Scalbar$ is the \emph{online} schedule where each segment executes up to its WCET and is suspended for up to its maximum suspension time.

Each \textbf{job} $J \in \Jbb$ has a certain starting and finishing time, denoting the first and the last time that a job is executed in a schedule.
We denote by $s_J$ and $f_J$ the starting and finishing time of job $J$ in the nominal schedule $\Scal$, respectively.
Moreover, we denote by $\sbar_J$ and $\fbar_J$ the starting and finishing time of job $J$ in the online schedule $\Scalbar$, respectively.
We denote by $r_J$ the release time of job $J$, i.e., $r_J \in \Rel_\tau$ if $\mTbb(J) = \tau$.

Similarly, each \textbf{computation segment} $\gamma \in \Cbb$ has a starting and a finishing time.
In the nominal schedule $\Scal$ they are denoted as $s_{\gamma}$ and $f_{\gamma}$, respectively, and in online schedule $\Scalbar$ they are denoted by $\sbar_{\gamma}$ and $\fbar_{\gamma}$, respectively.
We denote by $\ex(\gamma) \subset \Rbb$ the set of time points at which $\gamma$ is executed in $\Scal$.
Moreover, we denote by $C_\gamma \in \Rbb$ the total amount of time that the segment $\gamma$ is executed and by $S_\gamma$ the total amount of time that the segment $\gamma$ is suspended in $\Scal$.
By definition $C_\gamma = \mu(\ex(\gamma))$ holds for the Lebesgue measure $\mu$.
For the online schedule $\Scalbar$ we use the notation $\exbar(\gamma)$, $\Cbar_\gamma$ and $\Sbar_\gamma$ instead.

More specifically, let $J$ be a job of task $ \tau=(\Ex_\tau, \Rel_\tau )$
with $\Ex_\tau = (C_{\tau}^0, S_{\tau}^0, C_{\tau}^1, S_{\tau}^1, \ldots, S_{\tau}^{M_\tau-2}, C_{\tau}^{M_\tau-1})$, and let $(\gamma_0, \dots, \gamma_{M_\tau-1})$ denote the computation segments of job $J$.
Then in the nominal schedule $\Scal$

\begin{itemize}
	\item $\gamma_0$ is executed for $C_{\gamma_0} = C_{\tau}^0$ time units
	\item $\gamma_j$ is executed for $C_{\gamma_j} = C_{\tau}^j$ time units after suspending for $S_{\gamma_j} = S_{\tau}^{j-1}$ time units, $j \geq 1$
\end{itemize}
and in the online schedule $\Scalbar$
\begin{itemize}
	\item $\gamma_0$ is executed for $\Cbar_{\gamma_0} \in (0, C_{\tau}^0]$ time units
	\item $\gamma_j$ is executed for $\Cbar_{\gamma_j} \in (0, C_{\tau}^j]$ time units after suspending for $\Sbar_{\gamma_j} \in (0, S_{\tau}^{j-1}]$ time units, $j \geq 1$.
\end{itemize}

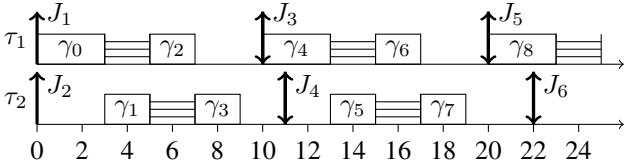
\begin{figure}[tb]
	\centering	
	\begin{tikzpicture}[yscale=0.4, xscale=0.3]
	\begin{scope}[shift={(0,2)}] 
		\taskname{$\tau_1$}
		
		\timeline{0}{26}{}
		
		\releases{0,10, 20}
		\deadlines{10, 20}
		
		\execname{0}{3}{$\gamma_0$}
		\execname{5}{7}{$\gamma_2$}
		\execname{10}{13}{$\gamma_4$}
		\execname{15}{17}{$\gamma_6$}
		\execname{20}{23}{$\gamma_8$}
		
		\susp{3}{5}
		\susp{13}{15}
		\susp{23}{25}

		\node[anchor=south west] at (0,1) {$J_1$};
		\node[anchor=south west] at (10,1) {$J_3$};
		\node[anchor=south west] at (20,1) {$J_5$};
	\end{scope}
	
	\begin{scope}[shift={(0,0)}] 
		\taskname{$\tau_2$}
		
		\timeline{0}{26}{}
		\labelling{0}{25}{2}{0}
		
		\releases{0,11,22}
		\deadlines{11,22}
		
		\execname{3}{5}{$\gamma_1$}
		\execname{7}{9}{$\gamma_3$}
		\execname{13}{15}{$\gamma_5$}
		\execname{17}{19}{$\gamma_7$}		
		
		\susp{5}{7}
		\susp{15}{17}

		\node[anchor=south west] at (0,0.6) {$J_2$};
		\node[anchor=south west] at (11,0.6) {$J_4$};
		\node[anchor=south west] at (22,0.6) {$J_6$};
	\end{scope}
	\end{tikzpicture}
	\caption{Example of a nominal schedule of a segmented self-suspending task set.}
	\label{fig:example_nominal}	
\end{figure}

\begin{exmpl}
	Figure~\ref{fig:example_nominal} demonstrates the nominal schedule of a task set with two segmented self-suspension periodic tasks,
	$\tau_1 = (\Ex_{\tau_1}, \Rel_{\tau_1}) = ((3, 2, 2), \\ (0, 10, 20, \dots))$ and $\tau_2 = (\Ex_{\tau_2}, \Rel_{\tau_2}) = ((2, 2, 2), (0, 11, 22, \dots))$.
	We have $\Tbb = \setof{\tau_1, \tau_2}$, $\Jbb = \setof{J_1, J_2, \dots}$, and $\Cbb = \setof{\gamma_0, \gamma_1, \dots}$. 
	The first job released by $\tau_1$, denoted as $J_1$, consists of two computation segments, $\gamma_0$ and $\gamma_2$.
	Segment
	$\gamma_0$ has a starting time $s_{\gamma_0} = 0$ and executes for $C_{\gamma_0} = 3$ time units during $\ex(\gamma_0) = [0,3)$, while $s_{\gamma_2} = 5$, $C_{\gamma_2} = 2$ and $\ex(\gamma_2) = [5,7)$ for $\gamma_2$.
	The maximum suspension time after executing $\gamma_0$, denoted as $S_{\gamma_0}$, is $2$.
	The starting time $s_{J_1}$ and finishing time $f_{J_1}$ of job $J_1$ are $0$ and $7$, respectively.
\end{exmpl}

\subsection{Release Jitter}\label{sec:task_model_jitter}

\emph{Release jitter} is the difference between the \emph{expected} and \emph{actual} release time of a job.
We denote by $\Jcal_{\tau}$ the \emph{maximum release jitter} among all jobs released by task $\tau$.
For periodic tasks with release jitter, the actual release time of a job $J$ may fall into the interval $[r_J, r_J + \Jcal_{\mTbb(J)}]$, where $r_J \in Rel_{\mTbb(J)}$ is the expected release time of $J$.
Release jitter may increase the response time of a job and interfere with the execution of other jobs, leading to potential deadline misses which jeopardize the timing guarantees of real-time systems, as shown in Figure~\ref{fig:release_jitter_anomaly}.

In this work, we focus exclusively on the release jitter at the job level.
For jobs released by the same task, they can experience different release jitters, but are bounded by the maximum release jitter of the task.
Release jitter on the segment level is not specifically addressed in this work, as it is expected to be encompassed within the suspension intervals.

The difference between actual and expected release time of a job can be modeled as self-suspension before its first computation segment.
That is, we
revise the execution behavior $\Ex_\tau$ of task $\tau$ introduced in Section~\ref{sec:task_model} to consider the release jitter as follows.
\begin{equation}
	\Ex_\tau = (S_{\tau}^{*}, C_{\tau}^0, S_{\tau}^0, C_{\tau}^1, S_{\tau}^1, \ldots, S_{\tau}^{M_\tau-2}, C_{\tau}^{M_\tau-1}), \nonumber
\end{equation} 
where $S_{\tau}^{*}$ represents the suspension interval for release jitter with the maximum length equal to $\Jcal_{\tau}$.
It is noteworthy that the suspension interval $S_{\tau}^{*}$ can be larger or equal to $0$ to allow cases without release jitter, while the other suspension intervals $S_{\tau}^j$ are assumed to be always larger than $0$. 
In our model, $\Rel_\tau$ is not affected by the release jitter, as it only specifies the expected release time.

With the revised task model, let $(\gamma_0, \dots, \gamma_{M_\tau-1})$ denote the computation segments of job $J$.
Then in the nominal schedule $\Scal$

\begin{itemize}	
	\item $\gamma_0$ is executed for $C_{\gamma_0} = C_{\tau}^0$ time units after suspending for $S_{\gamma_0} = S_{\tau}^{*}$ time units
	\item $\gamma_j$ is executed for $C_{\gamma_j} = C_{\tau}^j$ time units after suspending for $S_{\gamma_j} = S_{\tau}^{j-1}$ time units, $j \geq 1$
\end{itemize}
and in the online schedule $\Scalbar$
\begin{itemize}
	\item $\gamma_0$ is executed for $\Cbar_{\gamma_0} \in (0, C_{\tau}^0]$ time units after suspending for $\Sbar_{\gamma_0} \in [0, S_{\tau}^{*}]$ time units
	\item $\gamma_j$ is executed for $\Cbar_{\gamma_j} \in (0, C_{\tau}^j]$ time units after suspending for $\Sbar_{\gamma_j} \in (0, S_{\tau}^{j-1}]$ time units, $j \geq 1$
\end{itemize}

Throughout the remainder of this work, we employ the revised task model, specifically the segmented self-suspending task with release jitter, unless indicated otherwise.

\subsection{Segment-level Fixed-Priority Preemptive Scheduling}
\label{sec:scheduling}

In this work, we employ the \emph{Segment-level Fixed-Priority} (S-FP) scheduling mechanism for scheduling the segmented self-suspending tasks.
S-FP scheduling covers scheduling algorithms such as Earliest-Deadline-First (EDF) and Task-level Fixed-Priority (T-FP)~\cite{lehoczky89,DBLP:conf/rtss/Baker03}.
With S-FP, we formulate two different schedules, the \emph{nominal} schedule $\Scal$ and the \emph{online} schedule $\Scalbar$, as introduced in Section~\ref{sec:task_model}.

We make the following definitions for the nominal schedule $\Scal$.
A Segment-level Fixed-Priority (S-FP) scheduler chooses the segment with the highest priority among all segments in the ready queue for execution.
A segment is ``ready'' and inserted into the ready queue according to Definition~\ref{def:ready}.
If the chosen segment has a higher priority than the one currently being executed, the current executing segment is preempted and moved to the ready queue.

\begin{defn}\label{def:ready}
	A computation segment $\gamma \in \Cbb$ is \emph{ready} at time $t \in \Rbb$ in the nominal schedule $\Scal$ if:
	\begin{enumerate}
		\item there is remaining workload to be executed in $\gamma$ in $\Scal$;
		\item $\gamma$ is or has been \emph{released} at time $t$ in $\Scal$.
	\end{enumerate}
\end{defn}

\begin{defn}\label{def:release}
	Let $J \in \Jbb$ be a job of task $\tau$ consisting of segments $(\gamma_0, \dots, \gamma_{M_\tau-1})$.
	The first segment $\gamma_0$ is released once the first suspension interval has concluded, i.e., at time $r_{\gamma_0} = r_J + S_\tau^{*}$.
	A subsequent segment $\gamma_j$ is released as soon as the previous segment $\gamma_{j-1}$ finishes and the suspension time is consumed, i.e., at time $f_{\gamma_{j-1}} + S_\tau^{j-1}$.
	In general, we denote by $r_{\gamma}$ the release time of a segment $\gamma\in\Cbb$ in the nominal schedule $\Scal$. 
\end{defn}

With this definition of a segment being \emph{ready}, the segment-level fixed priority preemptive scheduler is \emph{work-conserving} on the segment-level, in the sense that whenever there are ready segments, the segment of the highest priority task is executed.
This leads to the following observation for the offline schedule.   

\begin{obs}\label{obs:finish_nominal_rel}
	Let $\gamma \in \Cbb$ be a segment.
	In $\Scal$, the segment $\gamma$ finishes at the lowest $t \in \Rbb$ such that
	\begin{equation}
		t \geq r_\gamma + W_\gamma(r_\gamma, t) + C_\gamma,
	\end{equation}
	where $W_\gamma(r_\gamma, t)$ is the total amount of time that higher priority segments are executed during the interval $[r_\gamma,t)$ in $\Scal$, i.e., 
	\begin{equation}\label{eq:def_W_nominal}
		W_\gamma(r_\gamma,t) := \mu\left( \bigcup_{\omega \gpi \gamma \in \Cbb} \ex(\omega) \cap [r_\gamma, t) \right)
	\end{equation}
\end{obs}

At all times during the interval $[r_\gamma, s_\gamma)$, segments with higher priorities than $\gamma$ are executed.
Hence, $W_\gamma(r_\gamma, t) = (s_\gamma - r_\gamma) + W_\gamma(s_\gamma, t)$ holds, and the observation can be reformulated as follows.

\begin{obs}\label{obs:finish_nominal_start}
	Let $\gamma \in \Cbb$ be a segment.
	In $\Scal$, the segment $\gamma$ finishes at the lowest $t \in \Rbb$ such that
	\begin{equation}
		t \geq s_\gamma + W_\gamma(s_\gamma, t) + C_\gamma.
	\end{equation}
\end{obs}

Similar to Definitions~\ref{def:ready} and \ref{def:release}, we define \emph{ready} and \emph{release time} in the \emph{online} schedule $\Scalbar$ as follows.

\begin{defn}\label{def:ready_online}
	A computation segment $\gamma \in \Cbb$ is \emph{ready} at time $t \in \Rbb$ in the online schedule $\Scalbar$ if:
	\begin{enumerate}
		\item there is remaining workload to be executed in $\gamma$ in $\Scalbar$;
		\item $\gamma$ is or has been \emph{released} at time $t$ in $\Scalbar$.
	\end{enumerate}
\end{defn}

\begin{defn}\label{def:release_online}
	Let $J \in \Jbb$ be a job of task $\tau$ consisting of segments $(\gamma_0, \dots, \gamma_{M_\tau-1})$.
	In the online schedule $\Scalbar$, the first segment $\gamma_0$ is released at time $\rbar_{\gamma_0} := r_J + \Sbar_{\gamma_0}$.
	A subsequent segment $\gamma_j$ is released at time $\rbar_{\gamma_j} := \fbar_{\gamma_{j-1}} + \Sbar_{\gamma_{j}}$.	
	In general, we denote by $\rbar_{\gamma}$ the release time of a segment $\gamma\in\Cbb$ in the online schedule $\Scalbar$.
\end{defn}

\section{Timing Anomalies and Enforcements}\label{sec:prob_def}

Scheduling tasks with self-suspending behavior can lead to timing anomalies.
\emph{Timing anomaly} refers to the response time increase of a job due to the reduction of the execution or suspension time of some other jobs.

Figure~\ref{fig:timing_anomaly} demonstrates an example of a timing anomaly.
Given two segmented self-suspending tasks $\tau_1$ and $\tau_2$, each with two computation segments.
The computation segments from $\tau_1$ have higher priorities than segments from $\tau_2$.
Figure~\ref{fig:timing_anomaly} (a) shows the nominal schedule generated based on the WCET and maximum suspension time of the segments.
If the suspension interval of $\tau_1$ finishes earlier, the second segment of $\tau_1$ preempts the first segment of $\tau_2$, leading to an increased response time from $\tau_2$, as shown in Figure~\ref{fig:timing_anomaly}~(b).

The release jitter of jobs can also lead to timing anomalies, as shown in Figure~\ref{fig:release_jitter_anomaly}.
Consider two tasks $\tau_1$ and $\tau_2$ with the maximum release jitter $\Jcal_{\tau_1} = 2$ and $\Jcal_{\tau_2} = 1$.
Figure~\ref{fig:release_jitter_anomaly} (a) shows a feasible schedule generated based on the maximum release jitter, the WCET, and maximum suspension time of the segments.
However, if the actual release jitter of $\tau_1$ is reduced to $1$, as shown in Figure~\ref{fig:release_jitter_anomaly}~(b), the response time of $\tau_2$ increases, causing a deadline misses.

Timing anomalies can affect the feasibility of a task set, as shown in the previous examples.
To be more specific, it is possible that a task set determined to be schedulable based on the maximum release jitter, the WCET, and the maximum suspension time of the segments can still have deadline violations during runtime due to timing anomalies.
Existing schedulability analyses account for timing anomalies by over-approximation.
To avoid that analytical pessimism, a treatment for eliminating the timing anomalies is essential.
To that end, different mechanisms to reduce the impact of timing anomalies have been developed in the literature:

\begin{itemize}
\item \emph{Period enforcer}~\cite{Raj:suspension1991} intends to
  apply a runtime rule so that \emph{``it forces tasks to behave like
    ideal periodic tasks from the scheduling point of view with no
    associated scheduling penalties''}, summarized in Section 4.3.1
  in the survey paper~\cite{suspension-review-jj}. However, Chen and
  Brandenburg~\cite{ChenBrandenburg-2016-LITES} show that
  \emph{``period enforcement~\cite{Raj:suspension1991} is not strictly
    superior (compared to the base case without enforcement) as it can
    cause deadline misses in self-suspending task sets that are
    schedulable without enforcement.''}
\item \emph{Release guard}~\cite{DBLP:conf/icdcs/SunL96} and
  \emph{release enforcement}~\cite{WC16-suspend-DATE} enforce
  the $j$-th computation segments of two consecutive jobs of a
  real-time task to be released with a guaranteed minimum
  inter-arrival time, summarized in Section 4.3.2 in the survey
  paper~\cite{suspension-review-jj}.
  Figure~\ref{fig:different_enforcement} provides a visual comparison 
  between the online schedules for task $\tau_i$ under two different 
  enforcement methods: \emph{release enforcement} as described in~\cite{WC16-suspend-DATE}, 
  and our proposed approach, \emph{segment release time enforcement}.
  With \emph{release enforcement}~\cite{WC16-suspend-DATE}, the inter-arrival 
  time of each segment is fixed, whereas in our proposed method, the release 
  time of each segment is determined according to the nominal schedule.
\item \emph{Slack enforcement}~\cite{LR:rtas10} creates execution
  enforcement by utilizing the available
  \emph{slack}. G\"unzel~and~Chen~\cite{GuenzelC21_RTSJournal} provide
  counterexamples, indicating that slack enforcement may provoke
  deadline misses and do not guarantee the same WCRT as without slack enforcement.
\end{itemize}

\begin{figure}[tb]
	\centering	
	\begin{subfigure}{\linewidth}
		\centering
		\begin{tikzpicture}[yscale=0.4, xscale=0.4]
			\begin{scope}[shift={(0,0)}] 
				\taskname{$\tau_i$}
				
				\timeline{0}{21}{}
				
				\releases{0, 10, 20}
				\deadlines{10, 20}
				
				\exec{0}{3}
				\exec{6}{9}
				\exec{10}{11}
				\exec{16}{18}				

				\draw[<-, dashed, very thick, color=red] (4,0) -- (4, 1.75);
				\draw[->, dashed, very thick, color=red] (6,0) -- (6, 1.75);
				\draw[<-, dashed, very thick, color=red] (14,0) -- (14, 1.75);
				\draw[->, dashed, very thick, color=red] (16,0) -- (16, 1.75);
				
				\susp{4}{6}
				\susp{14}{15}				
		
			\end{scope}
		\end{tikzpicture}
		\caption{\emph{Release enforcement} proposed in~\cite{WC16-suspend-DATE}}
	\end{subfigure}\\
	\begin{subfigure}{\linewidth}
		\centering
		\begin{tikzpicture}[yscale=0.4, xscale=0.4]
			\begin{scope}[shift={(0,0)}] 
				\taskname{$\tau_i$}
				
				\timeline{0}{21}{}
				\labelling{0}{20}{2}{0}
				
				\releases{0, 10, 20}
				\deadlines{10, 20}
				
				\exec{0}{3}
				\exec{5}{8}
				\exec{10}{11}
				\exec{13}{15}
				
				\susp{3}{5}
				\susp{11}{12}
				
				\draw[->, dashed, very thick, color=red] (5,0) -- (5, 1.75);				
				\draw[->, dashed, very thick, color=red] (13,0) -- (13, 1.75);
			\end{scope}
		\end{tikzpicture}
		\caption{\emph{Segment release time enforcement} in this work}
	\end{subfigure}

	\caption{The online schedules of task $\tau_i$ under different release time enforcement policies, assuming $\tau_i = ((3,2,3), D_i = 10)$.
	The red upward (downward) arrows indicate the enforced release time (deadline) of each computation segment.
	Note that the inter-arrival time of each segment is fixed in~\cite{WC16-suspend-DATE}.}
	\label{fig:different_enforcement}	
\end{figure}
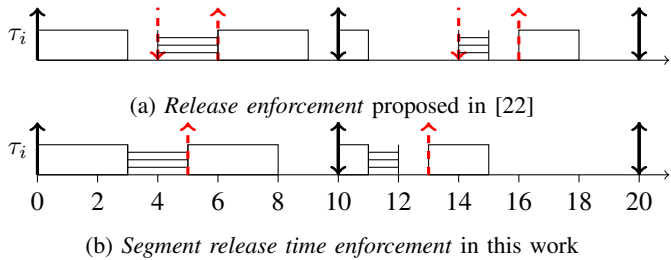

The period enforcer and slack enforcement pursue the
ultimate goal to completely \emph{ignore the self-suspension behavior} of
higher-priority
tasks, which would consequently avoid timing anomalies.
However, none of them achieves this ultimate goal, as shown by Chen~and~Brandenburg~\cite{ChenBrandenburg-2016-LITES} and
G\"unzel~and~Chen~\cite{GuenzelC21_RTSJournal}, and therefore they are not able to guarantee that timing anomalies are eliminated.
The release guard and release enforcement 
do not intend to ignore the self-suspension behavior of higher-priority
tasks or eliminate timing anomalies, but only aim for \emph{
easier schedulability analyses}.
The analyses are achieved by decoupling the segment release time from the finishing time of earlier segments.
Although not explicitly stated, such treatments avoid timing anomalies as a side effect.
However, the treatments do not sustain the Worst-Case Response Time (WCRT) of a task, meaning that the WCRT of a task with treatment may be much larger than the WCRT without treatment.

To the best of our knowledge, in this work we provide the first sustainable treatments that prevent timing anomalies for self-suspending tasks with release jitter.

\section{Treatments of Eliminating Timing Anomalies}\label{sec:treatments}

In this work, our objective is to eliminate timing anomalies without
negative impact on the worst-case response time when scheduling periodic tasks with segmented self-suspension behavior and release jitter.
To that end, we propose two treatments,
\emph{segment release time enforcement} and \emph{segment priority modification}.
In Section~\ref{sec:enforcement}, we introduce \emph{segment release time enforcement}, and prove that no timing anomaly can occur after applying this treatment.
However, \emph{segment release time enforcement} can lead to poor job response time in the average case, since it delays the segments artificially.
Therefore, we propose a \emph{segment priority modification} in Section~\ref{sec:total_order}, which eliminates timing anomalies 
without delaying the segment release time but by altering the segment priority based on the nominal schedule.
It is worth noting that the proposed treatments can handle tasks that release jobs according to \textit{any fixed release pattern}.

\subsection{Treatment 1: Enforcing the Release Time of Segments}\label{sec:enforcement}
Although the nominal schedule $\Scal$ for a segmented self-suspending task set is feasible, timing anomalies can still occur during runtime, as shown in the examples in Section~\ref{sec:prob_def}.
The cause of such timing anomalies is that a higher priority segment 
may start its execution before its release time in $\Scal$
due to the early completion and/or reduced suspension of a previous segment from the same job.
This higher priority segment blocks or preempts segments from lower priority jobs, therefore increasing their response times.
To eliminate timing anomalies, one method is to enforce the release time of the segments, such that all the segments start no earlier than their release time in the nominal schedule.

We note that our proposed enforcement is different from the release
guard~\cite{DBLP:conf/icdcs/SunL96} and release
enforcement~\cite{WC16-suspend-DATE}, presented in
Section~\ref{sec:prob_def}.  In release guard and release enforcement,
a computation segment of all jobs from a task have the same offset.
In our approach, the offset is computed for each job individually.

Our first treatment, \emph{segment release time enforcement}, works as follows.
If a segment $\gamma$ is ready at time $t$ according to Definition~\ref{def:ready_online} but the release time in the nominal schedule $r_\gamma$ is not reached, then the segment execution is delayed further until $r_\gamma$.
Formally, we \emph{redefine the term released} for the online schedule under segment release time enforcement as follows:

\begin{defn}\label{def:release_enforcement}	
	Let $J \in \Jbb$ be a job of task $\tau$ consisting of segments $(\gamma_0, \dots, \gamma_{M_\tau-1})$.
	In the online schedule $\Scalbar$ under \emph{segment release time enforcement},
	\begin{itemize}
		\item the first segment $\gamma_0$ is \emph{released} under two conditions: 1) after the release of job $J$ and the completion of the release jitter, and 2) once the release time in the nominal schedule $r_{\gamma_0}$ is reached.
		That is, segment $\gamma_0$ is released at time $\max(r_{J} + \Sbar_{\gamma_0}, r_{\gamma_0})$, where $0 \leq \Sbar_{\gamma_0} \leq \Jcal_{\mTbb(J)}$.		
		\item a subsequent segment $\gamma_j$ is \emph{released} as soon as the previous segment $\gamma_{j-1}$ finishes and the suspension time is consumed, and the release time in the nominal schedule $r_{\gamma_j}$ is reached. 
		That is, the segment $\gamma_j$ is released at time $\max(\fbar_{\gamma_{j-1}} + \Sbar_{\gamma_{j}}, r_{\gamma_j})$.		
	\end{itemize}
	In general, we denote by $\rbar^\enf_{\gamma}$ the release time of a segment $\gamma\in\Cbb$ in $\Scalbar$ under segment release time enforcement.
\end{defn}

In practice, we can enforce the segment release time by maintaining another queue.
All the segments which are supposed to be inserted into the ready queue are inserted to the new queue instead.
Only the segments that have a nominal starting time no lesser than the current time \(t\) can be moved to the ready queue.
Figure~\ref{fig:exec_states_enforcement} demonstrates the new execution states with \emph{segment release time enforcement}.
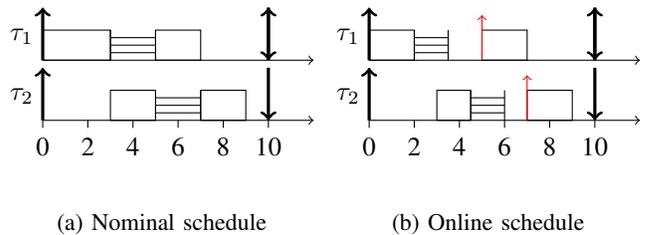
\begin{figure}[tb]
	\centering	
	\begin{subfigure}{.49\linewidth}
		\centering
		\begin{tikzpicture}[yscale=0.4, xscale=0.3]
			\begin{scope}[shift={(0,2)}] 
				\taskname{$\tau_1$}
				
				\timeline{0}{12}{}
				
				\releases{0,10}
				\deadlines{10}
				
				\exec{0}{3}
				\exec{5}{7}		
				
				\susp{3}{5}		
			\end{scope}
			
			\begin{scope}[shift={(0,0)}] 
				\taskname{$\tau_2$}
				
				\timeline{0}{12}{}
				\labelling{0}{10}{2}{0}
				
				\releases{0}
				\deadlines{10}
				
				\exec{3}{5}
				\exec{7}{9}		
				
				\susp{5}{7}
				
			\end{scope}
		\end{tikzpicture}
		\label{fig:exec_states_enforcement_a}
		\caption{Nominal schedule}
	\end{subfigure}%
	\begin{subfigure}{.49\linewidth}
		\centering
		\begin{tikzpicture}[yscale=0.4, xscale=0.3]
			\begin{scope}[shift={(0,2)}] 
				\taskname{$\tau_1$}
				
				\timeline{0}{12}{}
				
				\releases{0,10}
				\deadlines{10}
				
				\exec{0}{2}
				\exec{5}{7}		
				
				\susp{2}{3.5}		
				
				\draw[->, red] (5,0) -- (5,1.5);
			\end{scope}
			
			\begin{scope}[shift={(0,0)}] 
				\taskname{$\tau_2$}
				
				\timeline{0}{12}{}
				\labelling{0}{10}{2}{0}
				
				\releases{0}
				\deadlines{10}
				
				\exec{3}{4.5}			
				\exec{7}{9}		
				
				\susp{4.5}{6}
				
				\draw[->, red] (7,0) -- (7,1.5);
			\end{scope}
		\end{tikzpicture}
		\label{fig:exec_states_enforcement_b}
		\caption{Online schedule}
	\end{subfigure}	
	\caption{Example of the treatment \emph{segment release time enforcement}. A computation segment cannot start before its nominal release time (red arrow) even if the processor idles.}
	\label{fig:exec_states_enforcement}	
\end{figure}

Similar as in Definition~\ref{def:ready_online}, a segment is ready to be executed under segment release time enforcement if there is remaining workload to be executed and if the segment is released according to the previous definition.
In particular, with this definition of being ready, the segment-level fixed-priority preemptive scheduler under segment release time enforcement executes segments that are ready in a work-conserving manner, leading to the following observation.

\begin{obs}\label{obs:finish_online_srte}
	Let $\gamma \in \Cbb$ be a segment. 
	A segment $\gamma$ finishes in the online schedule with segment release time enforcement at the lowest $t \in \Rbb$ such that 
	\begin{equation}
		t \geq \rbar^\enf_\gamma + \Wbar_\gamma(\rbar^\enf_\gamma, t) + \Cbar_\gamma,
	\end{equation}
	where $\Wbar_\gamma(\rbar^\enf_\gamma, t)$ is the amount of time that higher priority segments are executed during the interval $[\rbar^\enf_\gamma,t)$ in the schedule $\Scalbar$ with segment release time enforcement, i.e., 
	\begin{equation}\label{eq:Wbar_enf}
		\Wbar_\gamma(\rbar^\enf_\gamma, t) 
		:= \mu\left( \bigcup_{\omega \gpi \gamma \in \Cbb} \exbar(\omega) \cap [\rbar^\enf_\gamma, t) \right)
	\end{equation}
\end{obs}

We now prove that with the treatment \emph{segment release time enforcement}, timing anomalies cannot occur in the online schedule $\Scalbar$. 
Intuitively, if the release time of a segment $\gamma \in \Cbb$ is fixed, the segment finishing time can only become larger if $\Wbar_\gamma$ is larger than $W_\gamma$.
However, since no segment release can be moved forward under the treatment, and if all previous segments finish no later than their finishing time in the nominal schedule, $\Wbar_\gamma$ cannot be larger than $W_\gamma$.
To make this proof formal, we start by rewriting $\Wbar_\gamma$ and $W_\gamma$ using the following lemma.

\begin{lemma}\label{lem:rel_enf_rewrite_W}
	Let $\gamma \in \Cbb$ be a segment.
	With segment release time enforcement, the following equations hold:
	\begin{align}
		\bigcup_{\omega \gpi \gamma \in \Cbb} \ex(\omega)
		& = \bigcup_{\omega \gpi \gamma \in \Cbb} [r_\omega, f_\omega) \label{eq:rel_enf_rewrite_W_nominal}
		\\
		\bigcup_{\omega \gpi \gamma \in \Cbb} \exbar(\omega)
		& = \bigcup_{\omega \gpi \gamma \in \Cbb} [\rbar^\enf_\omega, \fbar_\omega)  \label{eq:rel_enf_rewrite_W_online}
	\end{align}
\end{lemma}

\begin{proof}
	In the following we provide the proof for the nominal schedule $\Scal$ (Equation~\eqref{eq:rel_enf_rewrite_W_nominal}).
	The proof for the online schedule $\Scalbar$ (Equation~\eqref{eq:rel_enf_rewrite_W_online}) is analogous.
	
	$\subseteq$: 
	Since a segment can only be executed during the interval $[r_\omega, f_\omega)$, $\ex(\omega) \subseteq [s_\omega, f_\omega)$ holds for all segments $\omega \in \Cbb$.
	Therefore, $\bigcup_{\omega \gpi \gamma} \ex(\omega)
	\subseteq \bigcup_{\omega \gpi \gamma} [r_\omega, f_\omega)$ as well.
	
	$\supseteq$:
	Consider a segment $\omega \gpi \gamma$.
	Since the schedule is work-conserving,
	if the segment $\omega$ is not executed during $[r_\omega, f_\omega)$, then a segment with higher priority must be executed, i.e., 
	\begin{math}
		[r_\omega,f_\omega) 
		\subseteq \ex(\omega) \cup \bigcup_{\eta \gpi \omega} \ex(\eta).
	\end{math}
	Therefore, 
	\begin{math}
		\bigcup_{\omega \gpi \gamma} [r_\omega, f_\omega)
		\subseteq 
		\bigcup_{\omega \gpi \gamma} \left(\ex(\omega) \cup \bigcup_{\eta \gpi \omega} \ex(\eta)\right)
		= \bigcup_{\omega \gpi \gamma \in \Cbb} \ex(\omega)
	\end{math}
	which concludes the proof.
\end{proof}

With the previous lemma, we reformulate
$W_\gamma(r_\gamma,t)$ from Eq.~\eqref{eq:def_W_nominal} as \\ $\mu\left( \bigcup_{\omega \gpi \gamma \in \Cbb} [r_\omega, f_\omega) \cap [r_\gamma,t) \right)$ and $\Wbar_\gamma(\rbar^\enf_\gamma,t)$ from Eq.~\eqref{eq:Wbar_enf} as \\ $\mu\left( \bigcup_{\omega \gpi \gamma \in \Cbb} [\rbar^\enf_\omega, \fbar_\omega) \cap [\rbar^\enf_\gamma,t) \right)$.
This allows us to prove the following theorem which states that timing anomalies cannot occur under segment release time enforcement.

\begin{thm}\label{thm:enforcement}
	The finishing time of each segment 
	in the online schedule $\Scalbar$ with \emph{segment release time enforcement} is no larger than the finishing time in the nominal schedule $\Scal$, i.e., 
	$f_\gamma \geq \fbar_\gamma$ for all $\gamma \in \Cbb$.
\end{thm}

\begin{proof}
	Let $\Seg = (\gamma_0, \gamma_1, \dots)$ denote the list of all segments $\Cbb$ ordered by their finishing time in the nominal schedule, i.e., 
	$f_{\gamma_0} < f_{\gamma_1} < \dots$ holds.
	We consider the online schedule $\Scalbar$ obtained under segment release time enforcement.
	By induction over the segments in $\Seg$ we show that for each $\gamma_n$, $n=0,1,\dots$:
	\begin{itemize}
		\item[(i)] $\gamma_n$ is released at the same time in the online and in the nominal schedule, i.e., $\rbar^\enf_{\gamma_n} = r_{\gamma_n}$.
		\item[(ii)] The finishing time of $\gamma_n$ in the online schedule is no later than in the nominal schedule, i.e., $\fbar_{\gamma_n} \leq f_{\gamma_n}$.
	\end{itemize}
	
	\textbf{Base case (n=0):}
	Since $\gamma_0$ has the earliest finishing time in the nominal schedule, it must be the first segment of some job.
	By Definition~\ref{def:release_enforcement}, $\rbar^\enf_{\gamma_0} = \max(r_{\mJbb(\gamma_0)} + \Sbar_{\gamma_0} , r_{\gamma_0})$.
	Since $\Sbar_{\gamma_0} \leq \Jcal_{\mTbb(\mJbb(\gamma_0))}$ holds, we have $\rbar^\enf_{\gamma_0} = r_{\gamma_0}$, which implies that (i) holds.
	
	We prove (ii) by contradiction.
	By Observation~\ref{obs:finish_online_srte}, we know that $\fbar_{\gamma_0}$ is the smallest $t \in \Rbb$ such that $t\geq \rbar^\enf_{\gamma_0} + \Wbar_{\gamma_0}(\rbar^\enf_{\gamma_0}, t) + \Cbar_{\gamma} \overset{(i)}{=} r_{\gamma_0} + \Wbar_{\gamma_0}(r_{\gamma_0}, f_{\gamma_0}) + \Cbar_{\gamma}$.
	Assume that $f_{\gamma_0}<\fbar_{\gamma_0}$, then we have 
	\begin{equation}
		f_{\gamma_0} < r_{\gamma_0} + \Wbar_{\gamma_0}(r_{\gamma_0}, f_{\gamma_0}) + \Cbar_{\gamma_0}.
	\end{equation}	
	By the enforcement mechanism, $\rbar^\enf_{\omega} \geq r_\omega$ holds for all segments $\omega \in \Cbb$.
	Moreover, since $\gamma_0$ is the first segment in $\Seg$, it has the lowest finishing time in $\Scal$.
	Hence, 
	$[\rbar^\enf_\omega, \fbar_\omega) \cap [r_{\gamma_0}, f_{\gamma_0})
	\subseteq 
	[r_\omega, \fbar_\omega) \cap [r_{\gamma_0}, f_{\gamma_0})
	\subseteq 
	[r_\omega, f_{\gamma_0}) \cap [r_{\gamma_0}, f_{\gamma_0})
	\subseteq
	[r_\omega, f_\omega) \cap [r_{\gamma_0}, f_{\gamma_0})$.
	We obtain that
	\begin{equation}
		\bigcup_{\omega \gpi \gamma_0} [\rbar^\enf_\omega, \fbar_\omega) \cap [ r_{\gamma_0}, f_{\gamma_0})
		\subseteq 
		\bigcup_{\omega \gpi \gamma_0} [r_\omega, f_\omega) \cap [ r_{\gamma_0}, f_{\gamma_0})
	\end{equation}
	holds as well.
	Hence, $\Wbar_{\gamma_0}(r_{\gamma_0}, f_{\gamma_0}) \leq W_{\gamma_0}(r_{\gamma_0}, f_{\gamma_0})$ holds.
	We use that to obtain 
	\begin{align}
		f_{\gamma_0} 
		&< r_{\gamma_0} + \Wbar_{\gamma_0}(r_{\gamma_0}, f_{\gamma_0}) + \Cbar_{\gamma_0} \nonumber
		\\& \leq r_{\gamma_0} + W_{\gamma_0}(r_{\gamma_0}, f_{\gamma_0}) + \Cbar_{\gamma_0} \nonumber
		\\& \leq r_{\gamma_0} + W_{\gamma_0}(r_{\gamma_0}, f_{\gamma_0}) + C_{\gamma_0}. \nonumber
	\end{align}
	Since $r_{\gamma_0} + W_{\gamma_0}(r_{\gamma_0}, f_{\gamma_0}) + C_{\gamma_0} \leq f_{\gamma_0}$ holds by Observation~\ref{obs:finish_nominal_rel}, we obtain a contradiction. 
	This proves (ii).
	
	\textbf{Induction Step ($n -1 \mapsto n$):}
	We assume that (i) and (ii) hold for all segments in $\Seg_n := (\gamma_0, \gamma_1, \dots, \gamma_{n-1})$.
	In the following we show that (i) and (ii) hold for $\gamma_n$.
	
	For (i), if $\gamma_n$ is the first segment in its job $\mJbb(\gamma_n)$, then $\gamma_n$ is released at time $\rbar^\enf_{\gamma_n} = r_{\mJbb(\gamma_n)} = r_{\gamma_n}$ similar to the base case.
	If $\gamma_n$ is not the first segment of $\mJbb(\gamma_n)$, then denote by $\omega$ the segment of $\mJbb(\gamma_n)$ prior to $\gamma_n$.
	By definition $\rbar^\enf_{\gamma_n} \geq r_{\gamma_n}$. 
	$\rbar^\enf_{\gamma_n} > r_{\gamma_n}$ is only possible if $\fbar_\omega > f_{\omega}$.
	However, $\omega \in \Seg_n$ since $\omega$ finishes before $\gamma_n$.
	By induction, $\fbar_\omega > f_{\omega}$ is not possible.
	We conclude $\rbar^\enf_{\gamma_n} = r_{\gamma_n}$. 
	This proves (i).
	
	As in the base case we prove (ii) by contradiction and assume that $f_{\gamma_n} < \fbar_{\gamma_n}$.
	Similar to the base case, we obtain
	\begin{equation}
		f_{\gamma_n} < r_{\gamma_n} + \Wbar_{\gamma_n}(r_{\gamma_n}, f_{\gamma_n}) + \Cbar_{\gamma_n}.
	\end{equation}
	By the enforcement mechanism, 
	for all segments $\omega \in \Cbb$, $\rbar^\enf_\omega \geq r_\omega$ holds.
	Hence, 
	\begin{equation}\label{eq:rel_enf_ind_step_1}
		[\rbar^\enf_\omega, \fbar_\omega) \cap [r_{\gamma_n}, f_{\gamma_n}) \subseteq [r_\omega, \fbar_\omega) \cap [r_{\gamma_n}, f_{\gamma_n}).
	\end{equation}
	As in the base case, in the following we show that
	\begin{equation}\label{eq:rel_enf_ind_step_2}
		[r_\omega, \fbar_\omega) \cap [r_{\gamma_n}, f_{\gamma_n}) \subseteq [r_\omega, f_\omega) \cap [r_{\gamma_n}, f_{\gamma_n}).
	\end{equation}
	However, for the induction step we distinguish two cases:
	\textbf{If $f_\omega < f_{\gamma_n}$}, then we know that $\omega$ is in $\Seg_n$
	, and by induction $\fbar_\omega \leq f_\omega$ holds.
	Therefore, $[r_\omega, \fbar_\omega) \cap [r_{\gamma_n}, f_{\gamma_n}) \subseteq [r_\omega, f_\omega) \cap [r_{\gamma_n}, f_{\gamma_n})$.
	\textbf{If $f_\omega \geq f_{\gamma_n}$}, then 
	$[r_\omega, \fbar_\omega) \cap [r_{\gamma_n}, f_{\gamma_n}) 
	\subseteq [r_\omega, f_{\gamma_n}) \cap [r_{\gamma_n}, f_{\gamma_n})
	= [r_\omega, f_\omega) \cap [r_{\gamma_n}, f_{\gamma_n})$ holds as well.
	By applying Equations~\eqref{eq:rel_enf_ind_step_1} and~\eqref{eq:rel_enf_ind_step_2}, we obtain that
	\begin{equation}
		\bigcup_{\omega \gpi \gamma_n} [\rbar^\enf_\omega, \fbar_\omega) \cap [ r_{\gamma_n}, f_{\gamma_n})
		\subseteq 
		\bigcup_{\omega \gpi \gamma_n} [r_\omega, f_\omega) \cap [ r_{\gamma_n}, f_{\gamma_n}).
	\end{equation}
	Hence, $\Wbar_{\gamma_n}(r_{\gamma_n}, f_{\gamma_n}) \leq W_{\gamma_n}(r_{\gamma_n}, f_{\gamma_n})$ holds.
	Similar to the base case, we use that to obtain 
	\begin{align}
		f_{\gamma_n} 
		&< r_{\gamma_n} + \Wbar_{\gamma_n}(r_{\gamma_n}, f_{\gamma_n}) + \Cbar_{\gamma_n} \nonumber
		\\& \leq r_{\gamma_n} + W_{\gamma_n}(r_{\gamma_n}, f_{\gamma_n}) + \Cbar_{\gamma_n} \nonumber
		\\& \leq r_{\gamma_n} + W_{\gamma_n}(r_{\gamma_n}, f_{\gamma_n}) + C_{\gamma_n}. \nonumber
	\end{align}
	Since $r_{\gamma_n} + W_{\gamma_n}(r_{\gamma_n}, f_{\gamma_n}) + C_{\gamma_n} \leq f_{\gamma_n}$ holds by Observation~\ref{obs:finish_nominal_rel}, we obtain a contradiction, and (ii) is proven.
	This concludes the induction step and therefore proves the theorem.
\end{proof}

\subsection{Treatment 2: Modifying the Segment Priority}\label{sec:total_order}

\begin{figure*}[tb]
	\centering	
	\begin{subfigure}{.33\textwidth}
		\centering
		\begin{tikzpicture}[yscale=0.4, xscale=0.3]
			\begin{scope}[shift={(0,2)}] 
				\taskname{$\tau_1$}
				
				\timeline{0}{12}{}
				
				\releases{0,10}
				\deadlines{10}
				
				\exec{0}{3}		
				\exec{5}{7}

				\susp{3}{5}		
			\end{scope}
			
			\begin{scope}[shift={(0,0)}] 
				\taskname{$\tau_2$}
				
				\timeline{0}{12}{}
				\labelling{0}{10}{2}{0}
				
				\releases{0}
				\deadlines{10}
				
				\exec{3}{5}
				\exec{7}{9}

				\susp{5}{6}
			\end{scope}
		\end{tikzpicture}
		\label{fig:poor_performance_nominal}
		\caption{T-FP schedule (nominal)}
	\end{subfigure}%
	\begin{subfigure}{.33\textwidth}
	\centering
	\vspace{1em}
	\begin{tikzpicture}[yscale=0.4, xscale=0.3]
	\begin{scope}[shift={(0,2)}] 
		\taskname{$\tau_1$}
		
		\timeline{0}{12}{}
		
		\releases{0,10}
		\deadlines{10}
		
		\exec{0}{1}		
		\exec{5}{6}	
		
		\draw[->, red] (5,0) -- (5,1.5);
		
		\susp{1}{2}		
	\end{scope}
	
	\begin{scope}[shift={(0,0)}] 
		\taskname{$\tau_2$}
		
		\timeline{0}{12}{}
		\labelling{0}{10}{2}{0}
		
		\releases{0}
		\deadlines{10}
		
		\exec{1}{2}
		\exec{6}{7}		

		\draw[->, red] (6,0) -- (6,1.5);
		
		\susp{2}{2.5}
	\end{scope}
	\end{tikzpicture}
	\label{fig:poor_performance_a}
	\caption{Segment release time enforcement (online)}
	\end{subfigure}%
	\begin{subfigure}{.33\textwidth}
		\centering
		\begin{tikzpicture}[yscale=0.4, xscale=0.3]
		\begin{scope}[shift={(0,2)}] 
			\taskname{$\tau_1$}
			
			\timeline{0}{12}{}
			
			\releases{0,10}
			\deadlines{10}
			
			\exec{0}{1}
			\exec{2}{3}		
			
			\susp{1}{2}		
		\end{scope}
		
		\begin{scope}[shift={(0,0)}] 
			\taskname{$\tau_2$}
			
			\timeline{0}{12}{}
			\labelling{0}{10}{2}{0}
			
			\releases{0}
			\deadlines{10}
			
			\exec{1}{2}			
			\exec{3}{4}		
			
			\susp{2}{2.5}
		\end{scope}
		\end{tikzpicture}
		\label{fig:poor_performance_b}
		\caption{No treatment (online)}
		\end{subfigure}
	\caption{Schedules of the task set $\Tbb= \setof{\tau_1, \tau_2}$ under task-level fixed-priority scheduling where $\tau_1$ has higher priority than $\tau_2$.  Enforcing the segment release time leads to a longer per-job response time, compared to the schedule without treatment.}
	\label{fig:poor_performance}	
\end{figure*}
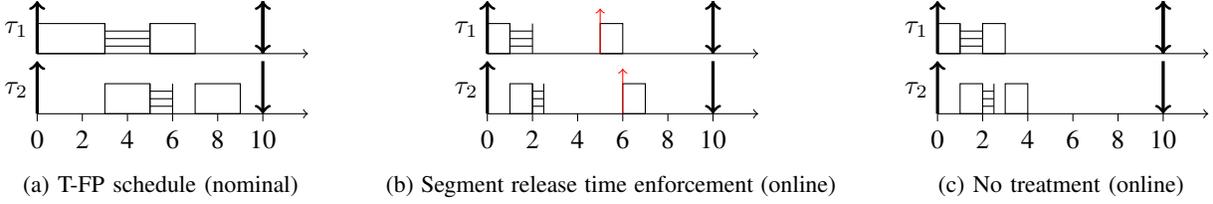

In Section~\ref{sec:enforcement}, we eliminate the timing anomalies and ensure that a feasible segmented self-suspending task set remains schedulable even with release jitter in online scheduling by enforcing the release time of the segments.
Although delaying the segment release has no negative impact on the worst-case behavior, 
this treatment may lead to poor average-case performance.
To demonstrate this, consider the case shown in Figure~\ref{fig:poor_performance}, where the WCETs and the maximum suspension times of segments are significantly larger than the actual execution and suspension times.
Figure~\ref{fig:poor_performance}~(a) demonstrates the nominal schedule, Figure~\ref{fig:poor_performance}~(b) shows the schedule with enforcement, and Figure~\ref{fig:poor_performance}~(c) shows the schedules without enforcement.
Compared to the schedule without enforcement, {enforcing the starting time} leads to a longer per-job response time, i.e., the time elapsed between the release and the completion of a job.
Therefore, a treatment without artificially delaying segments and without timing anomalies at the same time is desirable.

To that end we propose a treatment that modifies the segment priorities to eliminate timing anomalies.
In particular, we redefine the segment priorities according to their \textbf{finishing time} in the nominal schedule.
The rationale is that a segment with later finishing time should not be able to interfere a segment with an earlier finishing time. 
To distinguish original segment priorities and modified segment priorities, we call the modified priorities \emph{preference} instead.
More specifically, we say that a segment $\gamma \in \Cbb$ has a higher preference than $\omega \in \Cbb$
if $\gamma$ finishes earlier than $\omega$ in the nominal schedule.
We denote the preference as:
\begin{equation}
	\gamma >_P \omega \quad \ratio\Leftrightarrow \quad f_\gamma < f_\omega
\end{equation}
This leads to a total ordering of all segments in $\Cbb$.
In the online schedule, segments with higher preference are scheduled first.

For example, consider the system of Figure~\ref{fig:poor_performance}.
We denote by $\gamma^1_1$ and $\gamma^1_2$ the two segments of the first job of task $\tau_1$, and we denote by $\gamma^2_1$ and $\gamma^2_2$ the two segments of the first job of task $\tau_2$.
The total preference ordering is:
\begin{equation}
	\gamma^1_1 >_P \gamma^2_1 >_P \gamma^1_2 >_P \gamma^2_2
\end{equation}
Therefore, under the treatment with modified segment priorities, the online schedule will look exactly like the schedule without treatment in Figure~\ref{fig:poor_performance} (c).

Since in the online schedule, the segment priority is replaced with the segment preference, we can observe the following.

\begin{obs}\label{obs:finish_preference_rel}
	Let $\gamma \in \Cbb$ be a segment.
	The segment $\gamma$ finishes in the online schedule with priority modification at the lowest $t \in\Rbb$ such that
	\begin{equation}
		t \geq \rbar_\gamma + \Wbar_\gamma(\rbar_\gamma, t) + \Cbar_\gamma,
	\end{equation}
	where $\Wbar_\gamma(\rbar_\gamma, t)$ is the total amount of time that higher preference segments are executed during the interval $[\rbar_\gamma,t)$ in $\Scalbar$, i.e., 
	\begin{equation}
		\Wbar_\gamma(\rbar_\gamma,t) := \mu\left( \bigcup_{\omega >_P \gamma \in \Cbb} \exbar(\omega) \cap [\rbar_\gamma, t) \right).
	\end{equation}
\end{obs}

To prove that timing anomalies do not occur, we need to show that the interference from higher preference segments in the online schedule is not higher than the interference from higher priority segments in the nominal schedule.
To achieve this, we first prove the following key ingredient.

\begin{lemma}\label{lem:preference_key_ingredient}
	For all segments $\gamma \in \Cbb$, the interference in the nominal schedule is lower bounded by:
	\begin{equation}
		W_\gamma(s_\gamma, f_\gamma) \geq \sum_{\substack{\omega \in \Cbb\\ s_\gamma < f_\omega < f_\gamma }} C_\omega
	\end{equation}
\end{lemma}

\begin{proof}
	We know $W_\gamma(s_\gamma, f_\gamma) = \mu\left( \bigcup_{\omega \gpi \gamma} \ex(\omega) \cap [s_\gamma, f_\gamma) \right)$, by definition of $W_\gamma$ in Equation~\eqref{eq:def_W_nominal}.
	We first show that:
	\begin{equation}\label{eq:lem:preference_key_ingredient}
		\bigcup_{\omega \gpi \gamma} \ex(\omega) \cap [s_\gamma, f_\gamma) 
		\supseteq \bigcup_{\substack{\omega \in \Cbb\\ s_\gamma < f_\omega < f_\gamma }} \ex(\omega)
	\end{equation}
	To that end, consider an arbitrary $t \in \bigcup_{\substack{\omega \in \Cbb\\ s_\gamma < f_\omega < f_\gamma }} \ex(\omega)$.
	Then there exists $\omega \in \Cbb$ with $s_\gamma < f_\omega < f_\gamma$ such that $t \in \ex(\omega)$.
	Since $s_\gamma < f_\omega < f_\gamma$, we know that $\omega \neq \gamma$ and that $\omega$ is executed during $[s_\gamma, f_\gamma)$.
	Since no segment with lower priority than $\gamma$ can be executed during $[s_\gamma, f_\gamma)$, $\omega$ must have higher priority than $\gamma$, i.e., $\omega \gpi \gamma$.
	We know that every higher priority segment that is executed before $s_\gamma$ must finish before $\gamma$ can start. 
	Therefore, $\omega$ cannot be executed before $s_\gamma$ and $s_\omega \geq s_\gamma$ holds.
	Since $s_\gamma \leq s_\omega$ and $f_\omega < f_\gamma$, we know that $\ex(\omega) \subseteq [s_\gamma, f_\gamma)$.
	Hence, $t\in \ex(\omega) \cap [s_\gamma, f_\gamma)$.
	In conclusion, we have shown that there exists some $\omega \in \Cbb$ with $\omega \gpi \gamma$ such that $t\in \ex(\omega) \cap [s_\gamma, f_\gamma)$.
	This shows that $t \in \bigcup_{\omega \gpi \gamma} \ex(\omega) \cap [s_\gamma, f_\gamma)$, which proves Equation~\eqref{eq:lem:preference_key_ingredient}.

	By applying $\mu$ on both sides of Equation~\eqref{eq:lem:preference_key_ingredient}, we obtain that $W_\gamma(s_\gamma, f_\gamma) \geq \mu \left( \bigcup_{\substack{\omega \in \Cbb\\ s_\gamma < f_\omega < f_\gamma }} \ex(\omega) \right)$.
	Since only one segment can be executed at the same time, $\bigcup_{\substack{\omega \in \Cbb\\ s_\gamma < f_\omega < f_\gamma }} \ex(\omega)$ is a disjoint union, and we can apply $\mu$ on each $\ex(\omega)$ individually. 
	We obtain 
	$W_\gamma(s_\gamma, f_\gamma) \geq  \sum_{\substack{\omega \in \Cbb\\ s_\gamma < f_\omega < f_\gamma }} \mu(\ex(\omega))$ which proves this lemma.
\end{proof}

The previous lemma allows us to prove that no timing anomalies occur with the treatment that modifies segment priorities, as formulated in the following theorem.

\begin{thm}\label{thm:total_order}
	The finishing time of each segment in the online schedule $\Scalbar$ with segment preference instead of segment priorities is \emph{no larger} than the finishing time in the nominal schedule $\Scal$, i.e., $f_\gamma \geq \fbar_\gamma$ for all $\gamma \in \Cbb$.
\end{thm}

\begin{proof}
	In the proof of Theorem~\ref{thm:enforcement}, let $\Seg = (\gamma_0, \gamma_1, \dots)$ denote the list of all segments $\Cbb$ ordered by their finishing time in the nominal schedule, i.e., 
	$f_{\gamma_0} < f_{\gamma_1} < \dots$ holds.
	Please note that this ordering respects the segment preferences, i.e., $\gamma_0 >_P \gamma_1 >_P \dots$. 
	We consider the online schedule $\Scalbar$ obtained with segment preferences instead of segment priorities.
	By induction over the segments in $\Seg$ we show that for each $\gamma_n$, with $n=0,1,\dots$, the inequality $f_{\gamma_n} \geq \fbar_{\gamma_n}$ holds.
	
	\textbf{Base case ($n=0$):} We prove $f_{\gamma_0} \geq \fbar_{\gamma_0}$ by contradiction, i.e., we assume $f_{\gamma_0} < \fbar_{\gamma_0}$.
	By Observation~\ref{obs:finish_preference_rel}, we know that $\fbar_{\gamma_0}$ is the lowest $t \in \Rbb$ such that $t \geq \rbar_{\gamma_0} + \Wbar_{\gamma_0}(\rbar_{\gamma_0}, t) + \Cbar_{\gamma_0}$ holds.
	Since $f_{\gamma_0} < \fbar_{\gamma_0}$, we have 
	\begin{equation}
		f_{\gamma_0} < \rbar_{\gamma_0} + \Wbar_{\gamma_0}(\rbar_{\gamma_0}, f_{\gamma_0}) + \Cbar_{\gamma_0}.
	\end{equation}
	Since $\gamma_0$ has the earliest finishing time in $\Scal$ among all segments, it must be the first segment of its corresponding job $\mJbb(\gamma_0)$.
	The online release time of $\gamma_0$ is $\rbar_{\gamma_0} = r_{\mJbb(\gamma_0)} + \Jcal_{\gamma_0} \leq r_{\mJbb(\gamma_0)} + \Jcal_{\mTbb(\mJbb(\gamma_0))}  = r_{\gamma_0}$.
	Moreover, since $\gamma_0$ has the highest preference, $\Wbar_{\gamma_0}(\rbar_{\gamma_0}, f_{\gamma_0}) = \mu(\emptyset) = 0$.
	This leads us to	
	\begin{equation}
		\begin{split}
			f_{\gamma_0} 
			&< \rbar_{\gamma_0} + \Wbar_{\gamma_0}(\rbar_{\gamma_0}, f_{\gamma_0}) + \Cbar_{\gamma_0} \nonumber \\						
			&\leq r_{\gamma_0} + C_{\gamma_0} \nonumber
			\\& \leq r_{\gamma_0} + W_{\gamma_0}(r_{\gamma_0}, f_{\gamma_0}) + C_{\gamma_0}, \nonumber
		\end{split}
	\end{equation}
	Since $f_{\gamma_0} \geq r_{\gamma_0} + W_{\gamma_0}(r_{\gamma_0}, f_{\gamma_0}) + C_{\gamma_0}$ holds by Observation~\ref{obs:finish_nominal_rel}, we obtain a contradiction.
	This proves $f_{\gamma_0} \geq \fbar_{\gamma_0}$.
	
	\textbf{Induction Step ($n-1 \mapsto n$):}
	We assume that $\fbar_{\gamma_j} \leq f_{\gamma_j}$
	holds for all the previous segments $\gamma_j \in \Seg_n := (\gamma_0, \gamma_1, \dots, \gamma_{n-1})$.
	In the following we show that $\fbar_{\gamma_n} \leq f_{\gamma_n}$.
	To achieve this, we first prove that 
	\begin{equation}\label{eq:thm_pref_WbarleqW}
		\Wbar_{\gamma_n} ( s_{\gamma_n}, f_{\gamma_n})
		\leq W_{\gamma_n} ( s_{\gamma_n}, f_{\gamma_n}).
	\end{equation}
	By Lemma~\ref{lem:preference_key_ingredient}, we already know that 
	$W_{\gamma_n}(s_{\gamma_n}, f_{\gamma_n}) \geq \sum_{\substack{\omega \in \Cbb\\ s_{\gamma_n} < f_\omega < f_{\gamma_n} }} C_\omega$.
	Therefore, it is left to show that 
	\begin{equation}\label{eq:thm_pref_WbarleqsumC}
		\Wbar_{\gamma_n}(s_{\gamma_n}, f_{\gamma_n}) \leq \sum_{\substack{\omega \in \Cbb\\ s_{\gamma_n} < f_\omega < f_{\gamma_n} }} C_\omega.
	\end{equation}
	By definition of $\Wbar_{\gamma_n}$, we know that $\Wbar_{\gamma_n}(s_{\gamma_n}, f_{\gamma_n}) = \mu\left(\bigcup_{\omega >_P \gamma_n} \exbar(\omega) \cap [ s_{\gamma_n}, f_{\gamma_n} ) \right)$.
	In that definition, the condition $\omega >_P \gamma_n$ is equivalent to $f_\omega < f_{\gamma_n}$.
	Moreover, if $\exbar(\omega) \cap [ s_{\gamma_n}, f_{\gamma_n} ) \neq \emptyset$, then $\fbar_\omega > s_{\gamma_n}$ must hold.
	Since $\omega \in \Seg_n$, we have $f_{\omega} \geq \fbar_{\omega} > s_{\gamma_n}$ by induction.
	We obtain 
	\begin{align}
		\Wbar_{\gamma_n}(s_{\gamma_n}, f_{\gamma_n}) 
		&\leq \mu\left(\bigcup_{\substack{\omega \in \Cbb\\ \nonumber s_{\gamma_n} < f_\omega < f_{\gamma_n} }} \exbar(\omega) \cap [ s_{\gamma_n}, f_{\gamma_n} ) \right)
		\\ \nonumber & \leq \sum_{\substack{\omega \in \Cbb\\ s_{\gamma_n} < f_\omega < f_{\gamma_n} }} \Cbar_{\omega}
		\quad \leq \quad \sum_{\substack{\omega \in \Cbb\\ s_{\gamma_n} < f_\omega < f_{\gamma_n} }} C_{\omega}.
	\end{align}
	This proves Equation~\eqref{eq:thm_pref_WbarleqsumC} and therefore, also Equation~\eqref{eq:thm_pref_WbarleqW} is proven.
	
	We show $f_{\gamma_n} \geq \fbar_{\gamma_n}$ by contradiction, i.e., we assume that $f_{\gamma_n} < \fbar_{\gamma_n}$.
	By Observation~\ref{obs:finish_preference_rel}, we know that $\fbar_{\gamma_n}$ is the lowest $t \in \Rbb$ such that $t \geq \rbar_{\gamma_n} + \Wbar_{\gamma_n}(\rbar_{\gamma_n}, t) + \Cbar_{\gamma_n}$ holds.
	Since $f_{\gamma_n} < \fbar_{\gamma_n}$, we have 
	\begin{equation}
		f_{\gamma_n} < \rbar_{\gamma_n} + \Wbar_{\gamma_n}(\rbar_{\gamma_n}, f_{\gamma_n}) + \Cbar_{\gamma_n}.
	\end{equation}
	If $\gamma_n$ is the first segment of its job, then similar to the base case it is released at the release of the corresponding job and $\rbar_{\gamma_n} = r_{\gamma_n}$ holds.
	Otherwise, the segment $\xi \in \Cbb$ prior to $\gamma_n$ in its corresponding job $\mJbb(\gamma_n)$ is in $\Seg_n$.
	By induction $\fbar_\xi \leq f_\xi$, and therefore the segment $\gamma_n$ is released in $\Scalbar$ no later than in $\Scal$, i.e.,
	$\rbar_{\gamma_n} \leq r_{\gamma_n}\leq s_{\gamma_n}$. 
	We obtain
	\begin{align}
		f_{\gamma_n} 
		&< \rbar_{\gamma_n} + \Wbar_{\gamma_n}(\rbar_{\gamma_n}, f_{\gamma_n}) + \Cbar_{\gamma_n} \nonumber
		\\& \leq  \rbar_{\gamma_n} + (s_{\gamma_n} - \rbar_{\gamma_n}) + \Wbar_{\gamma_n}(s_{\gamma_n}, f_{\gamma_n}) + \Cbar_{\gamma_n} \nonumber
		\\& =s_{\gamma_n} + \Wbar_{\gamma_n}(s_{\gamma_n}, f_{\gamma_n}) + \Cbar_{\gamma_n} \nonumber
		\\& \leq s_{\gamma_n} + W_{\gamma_n}(s_{\gamma_n}, f_{\gamma_n}) + C_{\gamma_n}. \nonumber
	\end{align}
	Since $f_{\gamma_n} \geq s_{\gamma_n} + W_{\gamma_n}(s_{\gamma_n}, f_{\gamma_n}) + C_{\gamma_n}$ holds by Observation~\ref{obs:finish_nominal_start}, we obtain a contradiction.
	Hence, $f_{\gamma_n} \geq \fbar_{\gamma_n}$ holds.
	This concludes the induction step and therefore the theorem is proven. 
\end{proof}

\section{Application of Treatments}\label{sec:application}

In the previous section, the \emph{segment release time enforcement} and the \emph{segment priority modification} are introduced.
In this section we discuss how the treatments can be applied for a system of \emph{periodic}, \emph{synchronous}, segmented self-suspending real-time tasks with \emph{constrained deadlines} and release jitters.
More specifically, we assume that each task $\tau$ releases jobs according to its period $T_\tau >0$ starting at time $0$ (i.e., $\Rel_\tau = \setof{0, T_\tau, 2T_\tau, \dots}$) with a maximal release jitter $\Jcal_\tau \geq 0$,
and has a relative deadline $D_\tau \leq T_\tau$ (i.e., each job $J$ of task $\tau$ must finish until its absolute deadline $r_J + D_\tau$).

To apply the treatments, we follow a 2-step process:
\begin{itemize}
	\item \textbf{Step 1}: The \emph{nominal schedule} $\Scal$ is constructed and recorded offline, based on the maximal release jitter, the worst-case execution time and the maximum suspension time of a computation segment and a suspension interval, respectively.
	\item \textbf{Step 2}: The release times of all segments in $\Scal$ are used to enforce the online segment release times 
	for the \emph{segment release time enforcement}, and the finishing times are used to define the segment preference for the \emph{segment priority modification}.
	The online schedule $\Scalbar$ is generated according to the descriptions in Section~\ref{sec:treatments}.
\end{itemize}
Theorems~\ref{thm:enforcement} and~\ref{thm:total_order} show that under any of those treatments the finishing time of each segment in the online schedule $\Scalbar$ is upper bounded by the finishing time in the nominal schedule $\Scal$.
Therefore, the schedule with treatment is schedulable (i.e., each job finishes before its absolute deadline $\Scalbar$) if and only if the nominal schedule $\Scal$ is schedulable.

For periodic, synchronous tasks with constrained deadlines and release jitters, the nominal schedule $\Scal$ repeats every hyperperiod (i.e., the least common multiple of all task periods) if no deadline miss occurs in the first hyperperiod.
Therefore, it is sufficient to schedule only one hyperperiod and calculate the following finishing times of each segment accordingly.
Moreover, the simulation of one hyperperiod for Step 1 directly serves as an \emph{exact} schedulability test for scheduling under the treatments:
There are no deadline misses under schedule with treatment if and only if there are no deadline misses in the first hyperperiod of the nominal schedule $\Scal$.

\section{Evaluation}\label{sec:eval}

We compare the proposed treatments to state-of-the-art scheduling algorithms in terms of schedulability on synthetic task sets.
We first describe how the task sets are synthesized, and briefly introduce the comparing algorithms in Section~\ref{sec:eval_settings}.
In Section~\ref{sec:eval_schedulability}, we consider tasks without release jitter and compare the acceptance ratios of the algorithms under different task set configurations.
An approach based on a combination of scheduling algorithms is proposed to achieve a higher acceptance ratio.
In Section~\ref{sec:eval_jitter}, we demonstrate how release jitter affects the schedulability of the task sets.

The comparison of our proposed method with sporadic analyses in Section~\ref{sec:eval_schedulability} is reasoned by the absence of periodic analyses due to timing anomalies.
To the best of our knowledge, most existing researches for self-suspending tasks focus on the sporadic real-time task model.
Since the sporadic real-time task model is a \textit{behavior relaxation} against the periodic real-time task model, as demonstrated by von der Brueggen et al.~\cite{DBLP:conf/rtcsa/BruggenBCDR22} (stated in Appendix~\ref{sec:appendix_BR}), we compare our proposed treatments with the state-of-the-art approaches for the sporadic task model in the evaluation.

\subsection{Task Sets and Algorithms}\label{sec:eval_settings}

In our evaluation, we focus on segmented self-suspension periodic synchronous tasks with constrained deadlines.
The synthetic task sets were generated as follows.
First, we consider different total utilization settings of a task set, ranging from $0\%$ to $100\%$ in a $5\%$ step.
For each total utilization setting, we generated $100$ task sets, each with $10$ tasks $\setof{\tau_1, \dots, \tau_{10}}$.
Given the total utilization of a task set, we applied the Dirichlet-Rescale (DRS) algorithm~\cite{DBLP:conf/rtss/GriffinBD20} to determine the utilization $U_{\tau_i}$ of each individual task $\tau_i$.
$T_{\tau_i}$, the period of task $\tau_i$, was selected uniformly at random from a set of semi-harmonic periods $T_{\tau_i} \in \{1, 2, 5, 10, 20, 50, 100, 200, 1000\}$, which is used in automotive systems~\cite{DBLP:conf/rtns/BruggenUCF17,DBLP:conf/ecrts/HamannD0PW17,7509419}.
Each task $\tau_i$ has a relative deadline $D_{\tau} \leq T_{\tau}$.
With the utilization $U_{\tau_i}$ and period $T_{\tau_i}$, the total execution time of task $\tau_i$ was calculated accordingly, i.e., $C_{\tau_i} = U_{\tau_i} * T_{\tau_i}$.

Next, we divided the total execution time $C_{\tau_i}$ into the $M_{\tau_i}$ segments.
In our evaluation, $M_{\tau_i}$ was selected from the set $\{2 ~(\mathit{Rare}), 5~(\mathit{Moderate}), 8~(\mathit{Frequent})\}$ based on the configuration of the task set.
The number of suspension intervals was set to $M_{\tau_i} - 1$ accordingly.
The total suspension length of task $\tau_i$ was generated according to a uniform distribution in one of the following three ranges, as suggested in~\cite{WC16-suspend-DATE, vdBrueggen-RTCSA2017}:
\begin{itemize}
	\item Short suspension: $[0.01(T_{\tau_i} - C_{\tau_i}), 0.1(T_{\tau_i} - C_{\tau_i})]$
	\item Medium suspension: $[0.1(T_{\tau_i} - C_{\tau_i}), 0.3(T_{\tau_i} - C_{\tau_i})]$
	\item Long suspension: $[0.3(T_{\tau_i} - C_{\tau_i}), 0.6(T_{\tau_i} - C_{\tau_i})]$
\end{itemize}
Having the number of computation segments $M_{\tau_i}$, the total execution time $C_{\tau_i}$, and the total suspension length, we applied the DRS algorithm to determine the execution time of each computation segment and the length of each suspension interval, thus constructed the execution behavior $\Ex_{\tau_i}$ for task $\tau_i$ as defined in Section~\ref{sec:task_model}.

For task sets with release jitter, we further incorporated a maximum release jitter at the beginning of each tasks.
To determine the maximum release jitter for a task, we utilized the period of the shortest period task in a task set as the \emph{reference period}. 
We categorized release jitter severity into three levels: \textit{Minor}, \textit{Mild}, and \textit{Serious}, with the respective ranges specified as follows:
\begin{itemize}
	\item \textit{Minor}: $[1, 10\%]$ of the reference period
	\item \textit{Mild}: $[10, 20\%]$ of the reference period
	\item \textit{Serious}: $[20, 30\%]$ of the reference period
\end{itemize}
Based on the severity level, we generated the maximum release jitter for each task from a uniform distribution within the specified range.
With this approach, the maximum release jitter for tasks within a task set may vary, but are bounded within the specified range.

We considered the following segmented self-suspending scheduling algorithms:\footnote{The evaluation framework for self-suspending task systems, i.e., SSSEvaluation~\cite{DBLP:conf/rtss/GunzelTCBC21}, is applied for evaluating SCAIR-OPA, SCAIR-RM, and EDAGMF-OPA. The framework is available at~\url{https://github.com/tu-dortmund-ls12-rt/SSSEvaluation}.}
\begin{itemize}
	\item \textbf{NOM-EDF}: Our approach, the nominal schedules are generated using EDF scheduling.
	\item \textbf{NOM-RM}: Our approach, the nominal schedules are generated using RM scheduling.
	\item \textbf{SCAIR-OPA}~\cite{DBLP:conf/rtcsa/SchonbergerHBCC18}: A pseudo-polynomial time schedulability test under Audsley’s Optimal Priority assignment~\cite{Audsley01}.
	\item \textbf{SCAIR-RM}~\cite{DBLP:conf/rtcsa/SchonbergerHBCC18}: A pseudo-polynomial time schedulability test under RM priority assignment.
	\item \textbf{EDAGMF-OPA}~\cite{WC16-suspend-DATE}: A fixed-priority equal deadline assignment scheduling with Audsley’s Optimal Priority assignment.
\end{itemize}
Recall that the proposed treatments depend on information such as the nominal release times of segments and the total preference order in a nominal schedule.
\textbf{NOM-EDF} and \textbf{NOM-RM} generate a nominal schedule by simulating the execution of the given task set based on the WCET and maximum suspension time of the segments using EDF / RM scheduling, respectively.
As discussed in Section~\ref{sec:application}, our treatments derive a feasible schedule whenever the nominal schedule simulated over one hyperperiod is feasible.
Note that although we focus on synchronous periodic tasks to ensure the schedulability of the nominal schedule in our evaluation, \textbf{NOM-EDF} and \textbf{NOM-RM} can work on any task set with a repetitive release pattern, e.g., periodic tasks with different offsets, as long as the nominal schedule repeats.

\subsection{Schedulability under Different Task Set Configurations}\label{sec:eval_schedulability}

\begin{figure}[tbp]
	\centering	
	\includegraphics[width=0.95\linewidth]{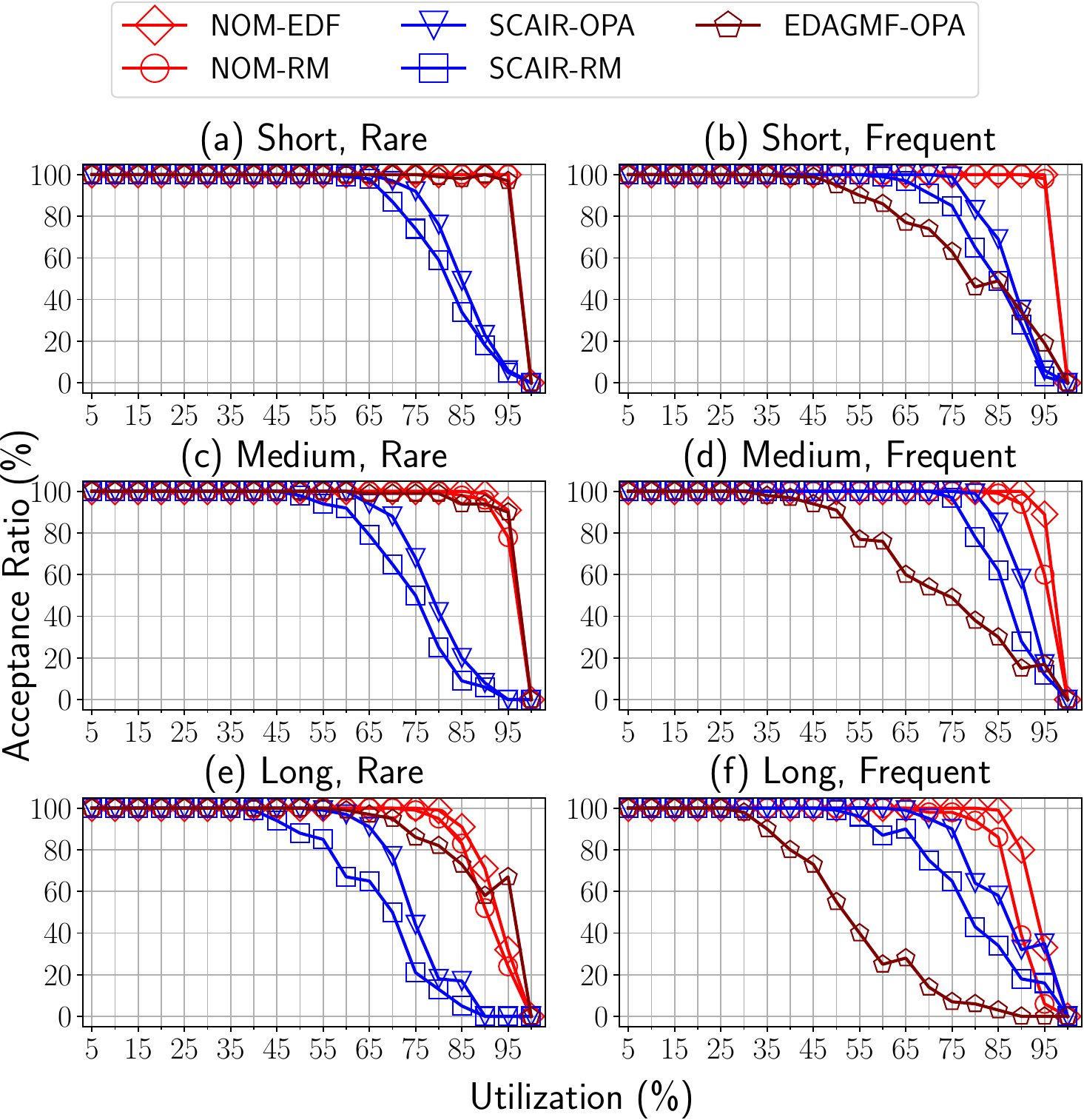}
	\caption{Acceptance ratio of the approaches under different task set configurations.  The results for the \emph{Moderate} task sets are provided in the Appendix.}
	\label{fig:sched_results_partial}	
\end{figure}

\begin{figure}[tp]
	\centering	
	\includegraphics[width=0.9\linewidth]{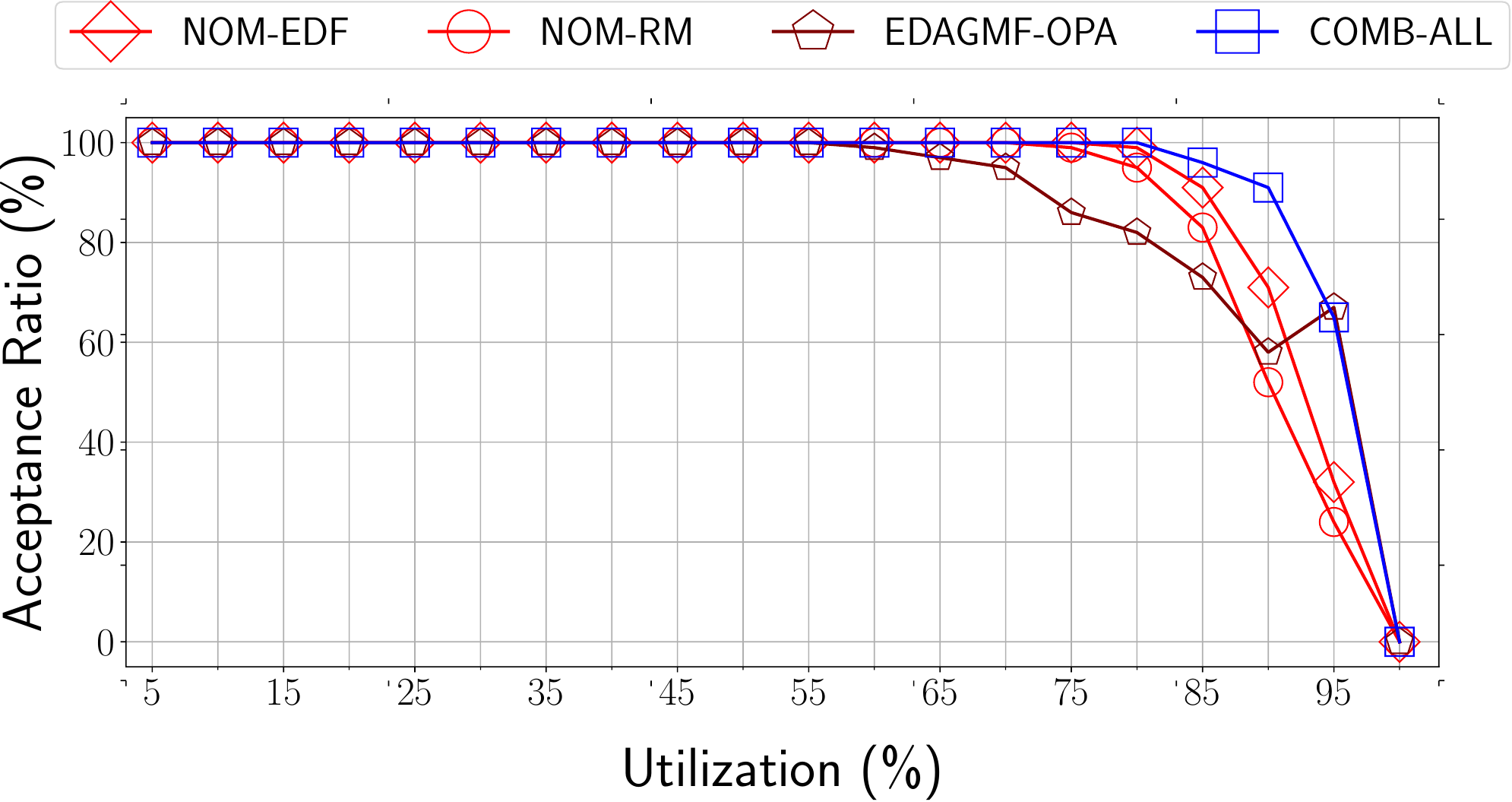}
	\caption{Acceptance ratios under the \textit{Long, Rare} configuration while considering the \textbf{COMB-ALL} approach.}
	\label{fig:sched_results_combine}	
\end{figure}

We compare the \emph{acceptance ratio} between the scheduling algorithms mentioned in Section~\ref{sec:eval_settings} using task sets \textbf{without} release jitter.
Figure~\ref{fig:sched_results_partial} demonstrates the results on task sets with different configurations, i.e., number of segments in a task and total suspension length.
We observe that in almost all the evaluated configurations, our \textbf{NOM-EDF} approach outperforms all the state of the arts.
The reason is that in order to eliminate timing anomalies for scheduling sporadic tasks, the existing methods over-approximate the WCRTs of tasks, which leads to overly pessimistic results.
The only exception appears in Figure~\ref{fig:sched_results_partial}~(e), where \textbf{EDAGMF-OPA} has the highest acceptance ratio among all algorithms when the total utilization reaches $95\%$ for task sets with long suspension intervals and only two segments.
We conclude that under certain configurations, priority assignment approaches such as \textbf{EDAGMF-OPA} can significantly improve the performance, i.e., schedulability, of fixed-priority scheduling.
Still, our proposed treatments remain high acceptance ratios under all the other configurations.

Although we achieved high acceptance ratios with \textbf{NOM-RM} and \textbf{NOM-EDF} in almost all the evaluated configurations,
we would like to point out that the proposed treatments do not bind to any specific scheduling algorithms for generating a nominal schedule.
Given a feasible nominal schedule of a segmented self-suspension periodic task set generated by any segment-level fixed-priority preemptive scheduling algorithm, \emph{segment release time enforcement} and \emph{segment priority modification} eliminate timing anomalies and guarantee the schedulability, as proven in Section~\ref{sec:treatments}.
Therefore, we propose a new approach \textbf{COMB-ALL}, which applies several scheduling algorithms to a task set, and returns a nominal schedule if the task set is feasible by any of these algorithms.
In our current design, we consider \textbf{NOM-EDF}, \textbf{NOM-RM}, and \textbf{EDAGMF-OPA} in \textbf{COMB-ALL} since these approaches in general outperformed the others in Figure~\ref{fig:sched_results_partial}.
Figure~\ref{fig:sched_results_combine} demonstrates the acceptance ratios of \textbf{COMB-ALL} and those of the three approaches individually under the same configuration in Figure~\ref{fig:sched_results_partial}~(e).
We observe that \textbf{COMB-ALL} has the highest acceptance ratio among the anomaly-free approaches.

\subsection{Impact of Release Jitter on Schedulability}\label{sec:eval_jitter}

\begin{figure}[tbp]
	\centering	
	\includegraphics[width=\linewidth]{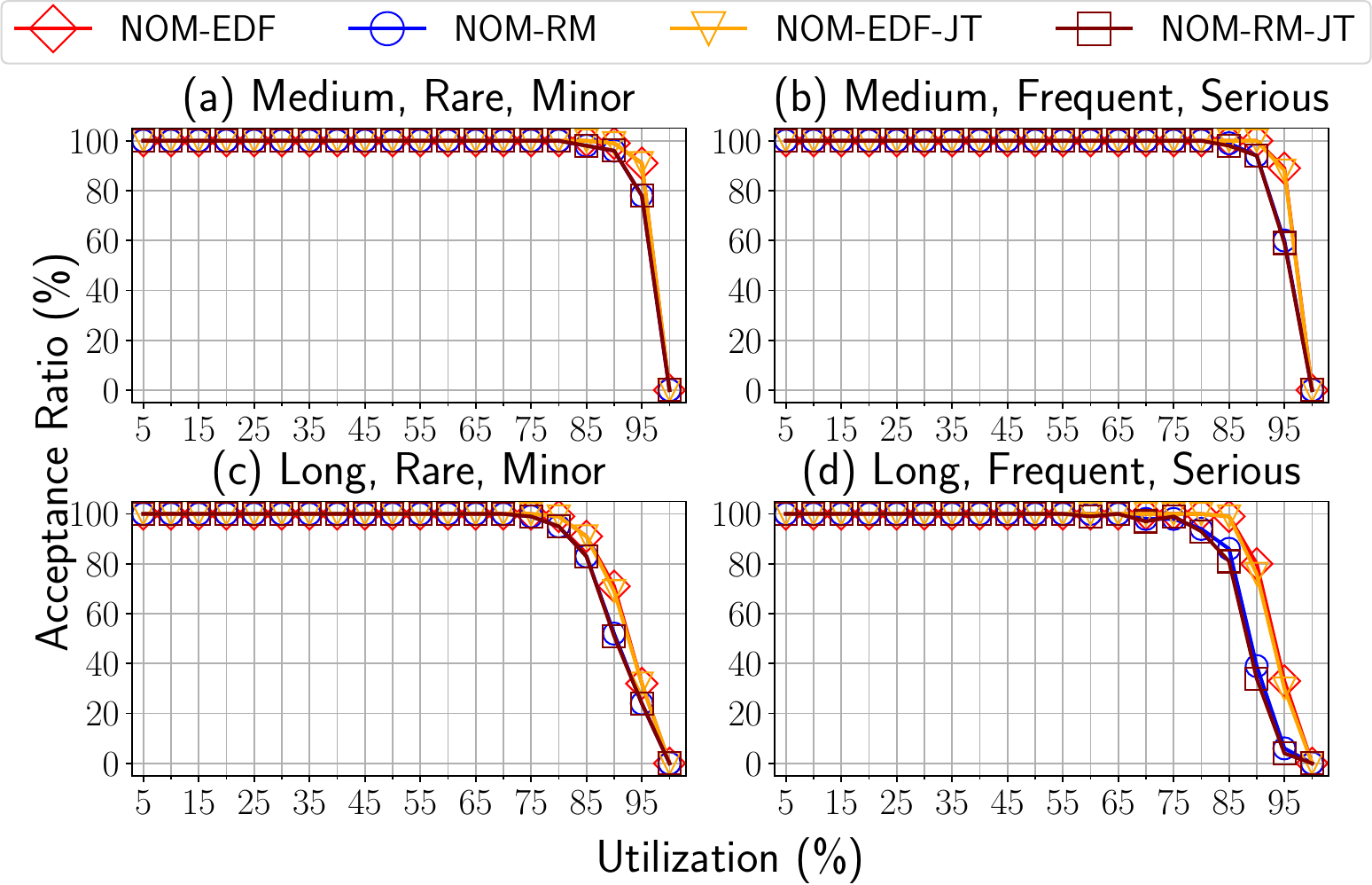}
	\caption{Acceptance ratios under different configurations with and without considering release jitter.}
	\label{fig:sched_results_jitter}	
\end{figure}

In this section, we assess the impact of release jitter on the acceptance ratio of nominal schedules.
Recall that the proposed treatments for eliminating timing anomalies rely on details provided in the nominal schedule.
Generating a nominal schedule without accounting for release jitter assumes that the first segment of every job is released simultaneously with the job itself.
Compared to the nominal schedules that neglect release jitter, the nominal schedules considering release jitter are expected to have lower acceptance ratios, depending on the severity of the release jitter.

Figure~\ref{fig:sched_results_jitter} demonstrates the acceptance ratios of the approaches under different configurations.
\emph{NOM-EDF} and \emph{NOM-RM} represent the approaches without considering release jitter, while \emph{NOM-EDF-JT} and \emph{NOM-RM-JT} incorporate the maximum release jitter of tasks during nominal schedule generation.
We observe that the acceptance ratios of \emph{NOM-EDF-JT} and \emph{NOM-RM-JT} are lower than those of \emph{NOM-EDF} and \emph{NOM-RM} only when the utilization is high.
The difference becomes more significant when the release jitter is more severe.

Although considering the maximum release jitter of tasks may marginally reduce the acceptance ratio, the proposed treatments guarantee that there are no timing anomalies in the online schedule.
On the other hand, disregarding the release jitter may yield overly optimistic results, potentially resulting in timing anomalies in the online schedule, as depicted in Figure~\ref{fig:release_jitter_anomaly}.

\section{Implementation on RTEMS}
\label{sec:implementation}

We implemented a Segment-Level Fixed-Priority (S-FP) scheduling mechanism on RTEMS, an open-source RTOS, to demonstrate the applicability of the treatment \emph{segment priority modification}.
In Section~\ref{sec:impl_details}, we introduce the key APIs for implementing the S-FP scheduling.
We then showcase the validity of the treatment on RTEMS with a working example in Section~\ref{sec:impl_working_example}.

\subsection{Implementation of S-FP Scheduling}\label{sec:impl_details}

\begin{figure}[tb]
	\centering	
	\newcommand{\changeprio}[1]{\draw[->, red, thick] (#1,1.8) -- (#1,1);}
	\begin{subfigure}{.33\linewidth}
		\centering
		\begin{tikzpicture}[yscale=0.4, xscale=0.225]
			\begin{scope}[shift={(0,2)}] 
				\taskname{$\tau_1$}
				\timeline{0}{10}{}
				\releases{0}
				
				\execname{0}{2}{$1$}
				\changeprio{2}
				\execname{5}{6}{$3$}
				\susp{6}{9}
			\end{scope}
			
			\begin{scope}[shift={(0,0)}] 
				\taskname{$\tau_2$}
				\timeline{0}{10}{}
				\labelling{0}{9}{2}{0}
				\releases{0}
				
				\execname{2}{5}{$2$}
			\end{scope}
		\end{tikzpicture}
		\caption{Before}
		\label{fig:slfp_priority_a}
	\end{subfigure}%
	\hfill
	\begin{subfigure}{.33\linewidth}
		\centering
		\begin{tikzpicture}[yscale=0.4, xscale=0.225]
			\begin{scope}[shift={(0,2)}] 
				\taskname{$\tau_1$}
				\timeline{0}{10}{}
				\releases{0}
				
				\execname{0}{3}{$1$}
				\susp{3}{5}
				\execname{5}{6}{$1$}
				\changeprio{6}
				\execname{7}{9}{3}
			\end{scope}
			
			\begin{scope}[shift={(0,0)}] 
				\taskname{$\tau_2$}
				\timeline{0}{10}{}
				\labelling{0}{9}{2}{0}
				\releases{0}
				
				\execname{3}{5}{$2$}
				\execname{6}{7}{$2$}	
			\end{scope}
			\draw[->, dashed, red] (5,1) -- (5,2);
		\end{tikzpicture}

		\caption{After}
		\label{fig:slfp_priority_b}
	\end{subfigure}%
	\hfill
	\begin{subfigure}{.33\linewidth}
		\centering
		\begin{tikzpicture}[yscale=0.4, xscale=0.225]
			\begin{scope}[shift={(0,2)}] 
				\taskname{$\tau_1$}
				\timeline{0}{10}{}
				\releases{0}
				
				\execname{0}{3}{$1$}
				\susp{3}{5}
				\execname{6}{9}{$3$}
			\end{scope}
			
			\begin{scope}[shift={(0,0)}] 
				\taskname{$\tau_2$}
				\timeline{0}{10}{}
				\labelling{0}{9}{2}{0}
				\releases{0}
				
				\execname{3}{6}{$2$}
				\changeprio{4.5}
			\end{scope}
			\draw[->, dashed, red] (5,1) -- (5,2);
		\end{tikzpicture}

		\caption{During}
		\label{fig:slfp_priority_c}
	\end{subfigure}
	\caption{Impact of the time to perform priority modification. The red arrow indicates the time point that the priority is modified. Only option (c) leads to the desired behavior.}
	\label{fig:slfp_priority}	
\end{figure}
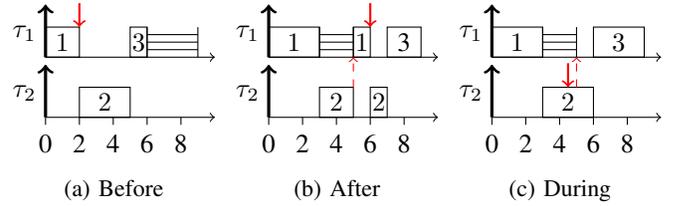

\begin{figure*}[tb]
	\centering	
	\begin{subfigure}{.33\linewidth}
		\centering
		\begin{tikzpicture}[yscale=0.4, xscale=0.225]
			\begin{scope}[shift={(0,2)}] 
				\taskname{$\tau_1$}
				\timeline{0}{14}{}
				\releases{0,12}
				
				\execname{0}{3}{1}
				\susp{3}{8}
				\execname{8}{11}{1}
			\end{scope}
			
			\begin{scope}[shift={(0,0)}] 
				\taskname{$\tau_2$}
				\timeline{0}{14}{}				
				\releases{0,6,12}
				
				\execname{3}{4}{2}
				\execname[fill=cyan]{6}{7}{2}
			\end{scope}
			
			\begin{scope}[shift={(0,-2)}] 
				\taskname{$\tau_{\mathit{sus}}$}
				\timeline{0}{14}{}
				\labelling{0}{12}{2}{0}
				\releases{0,12}
				
				\execname{4}{6}{}
				\execname{7}{8}{}
			\end{scope}
			\draw[->, dashed, red] (8,-1) -- (8,2);
		\end{tikzpicture}
		\caption{T-FP WCET}
	\end{subfigure}%
	\begin{subfigure}{.33\linewidth}
		\centering
		\begin{tikzpicture}[yscale=0.4, xscale=0.225]
			\begin{scope}[shift={(0,2)}] 
				\taskname{$\tau_1$}
				\timeline{0}{14}{}
				\releases{0,12}
				
				\execname{0}{1}{1}
				\susp{1}{5}
				\execname{5}{8}{1}
			\end{scope}
			
			\begin{scope}[shift={(0,0)}] 
				\taskname{$\tau_2$}
				\timeline{0}{14}{}
				\releases{0,6,12}
				
				\execname{1}{2}{2}
				\execname[fill=cyan]{8}{9}{2}
			\end{scope}
			\begin{scope}[shift={(0,-2)}] 
				\taskname{$\tau_{\mathit{sus}}$}
				\timeline{0}{14}{}
				\labelling{0}{12}{2}{0}
				\releases{0,12}
				
				\execname{2}{5}{}		
			\end{scope}
			\draw[->, dashed, red] (5,-1) -- (5,2);
		\end{tikzpicture}
		\caption{T-FP with early completion}
	\end{subfigure}%
	\begin{subfigure}{.33\linewidth}
		\centering
		\begin{tikzpicture}[yscale=0.4, xscale=0.225]
			\begin{scope}[shift={(0,2)}] 
				\taskname{$\tau_1$}
				\timeline{0}{14}{}
				\releases{0,12}
				
				\execname{0}{1}{1}
				\susp{1}{5}
				\execname{5}{6}{4}
				\execname{7}{9}{4}
			\end{scope}
			
			\begin{scope}[shift={(0,0)}] 
				\taskname{$\tau_2$}
				\timeline{0}{14}{}
				\releases{0,6,12}
				
				\execname{1}{2}{2}	
				\execname[fill=cyan]{6}{7}{3}
			\end{scope}
			\begin{scope}[shift={(0,-2)}] 
				\taskname{$\tau_{\mathit{sus}}$}
				\timeline{0}{14}{}
				\labelling{0}{12}{2}{0}
				\releases{0,12}
				
				\execname{2}{5}{}
			\end{scope}
			\draw[->, dashed, red] (5,-1) -- (5,2);
		\end{tikzpicture}
		\caption{S-FP with early completion}
	\end{subfigure}	
	\caption{The working example implemented in RTEMS.}
	\label{fig:validation_cases_RTEMS}
\end{figure*}
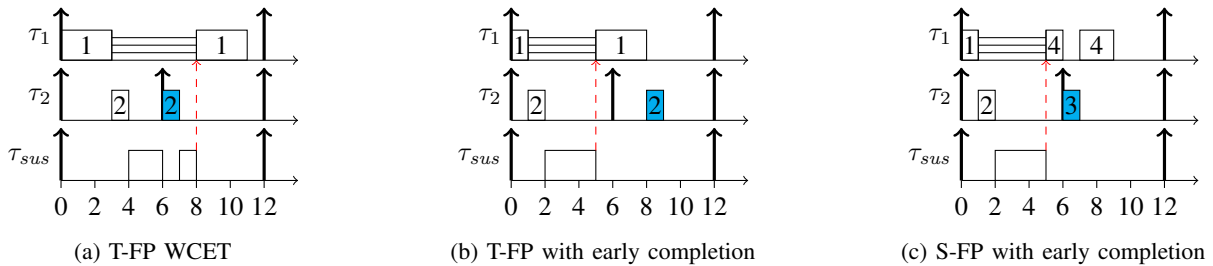

There are three major functionalities to be taken into consideration while designing the Segment-Level Fixed-Priority (S-FP) scheduling mechanism on RTEMS:
1) self-suspension of a task, 2) resuming a self-suspended task, and 3) modifying task priority.
In our current design, we introduce a middle layer which wraps all the required functions provided by RTEMS at the API layer without modifying the underlying kernel.
To perform self-suspension, the current executing task calls the function \texttt{rtems\_task\_suspend()} with its own task id \texttt{RTEMS\_SELF} at the end of a segment execution, except for the last segment.
Since a task cannot resume itself after suspension, the resume of a suspended task must be triggered by other sources, e.g., another task.

In the S-FP scheduling, segments from the same task are allowed to have different priorities.
However, RTEMS only supports task-level priority assignment in the current version.
Therefore, we adjusted the priority for each segment by calling the function \texttt{rtems\_task\_set\_priority()}.

There are three possible moments for modifying the priority: \textit{Before}, \textit{After}, and \textit{During} the suspension of the previous segment, as shown in Figure~\ref{fig:slfp_priority}.
If the priority is modified \textit{before} the suspension starts, i.e., during the execution of the previous segment, the remaining execution of the previous segment can be preempted by another job (Figure~\ref{fig:slfp_priority}~(a)).
Alternatively, if the priority is modified \textit{after} the segment starts, it can incur unexpected preemptions (Figure~\ref{fig:slfp_priority}~(b)).
This is due to the gap between calling the function \texttt{rtems\_task\_set\_priority()} and the priority modification is issued.
Therefore, the ideal solution is to perform priority modification \textit{during} the suspension, as shown in Figure~\ref{fig:slfp_priority}~(c).

Considering the priority modification during suspension, we introduce a customized resume mechanism.
Given the priority of each segment as inputs, we keep a lookup table in the middle layer.
Every time before a suspended task is about to be resumed, the controlling task first calls the function \texttt{rtems\_task\_set\_priority()} to assign the new priority to the task according to the lookup table, then it calls \texttt{rtems\_task\_resume()} with the id of the task to be resumed.

\subsection{Working Example}\label{sec:impl_working_example}

We validated the proposed treatment, \emph{segment priority modification}, on RTEMS with our S-FP implementation using the following example.
Given a taskset $\Tbb=\setof{\tau_1, \tau_2}$ of two tasks to be scheduled on a uniprocessor system, where $\tau_1 = (\Ex_{\tau_1}, \Rel_{\tau_1}) = ((3, 5, 3), (0, 12, 24, \dots))$, and $\tau_2 = (\Ex_{\tau_2}, \Rel_{\tau_2}) = ((1), (0, 6, 12, \dots))$.
Since a task cannot resume itself after suspension, we add one additional lowest priority task $\tau_{\mathit{sus}} = (\Ex_{\tau_3}, \Rel_{\tau_3}) = ((3), (0, 12, 24, \dots))$ which resumes $\tau_1$ when it finishes execution.
We consider three scenarios: (a) \textbf{T-FP WCET}, (b) \textbf{T-FP with early completion}, and (c) \textbf{S-FP with early completion}.
In \textbf{T-FP WCET} and \textbf{T-FP with early completion}, 
the tasks are scheduled using task-level fixed priority scheduling.
Each segment is executed to its WCET, and then suspended for the maximum length of the suspension interval in \textbf{T-FP WCET}.
On the other hand, the segments and the suspension interval can finish earlier in \textbf{T-FP with early completion}.
In this scenario, we decrease the execution time of the first segment in $\tau_1$.
In \textbf{S-FP with early completion}, the tasks are scheduled using task-level fixed priority scheduling.
The priorities of the segments follow the treatment \emph{segment priority modification}, which use the schedule generated by task-level fixed priority scheduling as the nominal schedule.

Figure~\ref{fig:validation_cases_RTEMS} demonstrates the schedules generated in the three scenarios.
The numbers in the segments are their priorities.
A lower number indicates a higher priority.
We observe that in Figure~\ref{fig:validation_cases_RTEMS}~(b), the first segment of $\tau_1$ finishes earlier, which delays the second job of $\tau_2$.
With priorities generated from the treatment \emph{segment priority modification}, the second job of $\tau_2$ is not affected, i.e., no timing anomalies occur.

\section{Conclusion and Future Work}\label{sec:conclusion}

For tasks with self-suspending behavior, providing timing guarantees is challenging due to timing anomalies, i.e., the reduction of execution or suspension time of some jobs enlarges the response time of another job.
In this paper, we propose two treatments, \emph{segment release time enforcement} and \emph{segment priority modification}, for scheduling segmented self-suspension periodic tasks without any risk of timing anomalies.
Given a nominal schedule generated based on the WCET and the maximum suspension time of segments, \emph{segment release time enforcement} eliminates timing anomalies by enforcing the release time of each segment to be no earlier than its nominal release time.
On the other hand, \emph{segment priority modification} maintains the total order of the segments in the nominal schedule to prevent timing anomalies.

In our evaluation, we compared the proposed treatments to state-of-the-art scheduling algorithms in terms of schedulability.
The results on synthetic task sets show that our proposed treatments achieve the highest acceptance ratio under almost all scenarios compared to the state of the art.
For the scenarios involving release jitter, the schedulability of the proposed treatments does not degrade significantly compared to ignoring the release jitter.
We also depict how to realize the segment-level fixed-priority scheduling mechanism on RTEMS, an open-source RTOS, and showcase the validity of the treatment \emph{segment priority modification} with an example.

\section*{Acknowledgment}
This result is part of a project (PropRT) that has received funding from the European Research Council (ERC) under the European Union’s Horizon 2020 research and innovation programme (grant agreement No. 865170).
This paper has been (partly) supported by Deutsche Forschungsgemeinschaft (DFG) Sus-Aware (Project No. 398602212), the Federal Ministry of Education and Research (BMBF) in the course of the project 6GEM under the funding reference 16KISK038.

\IEEEtriggeratref{45}
\bibliographystyle{elsarticle-num}
\bibliography{real-time}

\newpage

\appendices

\section{Behaviour Relaxation}
\label{sec:appendix_BR}

von~der~Br\"uggen~et~al.~\cite{DBLP:conf/rtcsa/BruggenBCDR22} provide
a ``\emph{systematic view of the value of special cases and the possible drawbacks of placing too much emphasis on generalization in real-time systems research}''. 
In their definition,
\begin{quote}
	\emph{Behaviour relaxation is a generalization that enables additional behaviour to be modelled. More specifically, model $A$ is a behaviour relaxation of model $B$ if the runtime behaviours of the systems described by model $A$ are strict supersets of the runtime behaviours of the corresponding systems described by model $B$.}
\end{quote}
They demonstrate that the sporadic real-time task model is a behavior relaxation against the periodic real-time task model and the dynamic self-suspension model is also a behavior relaxation against the segmented self-suspension model.

\section{Figures for the Evaluation}\label{app:figures}

\begin{figure}[tb]
	\centering	
	\includegraphics[width=0.95\textwidth]{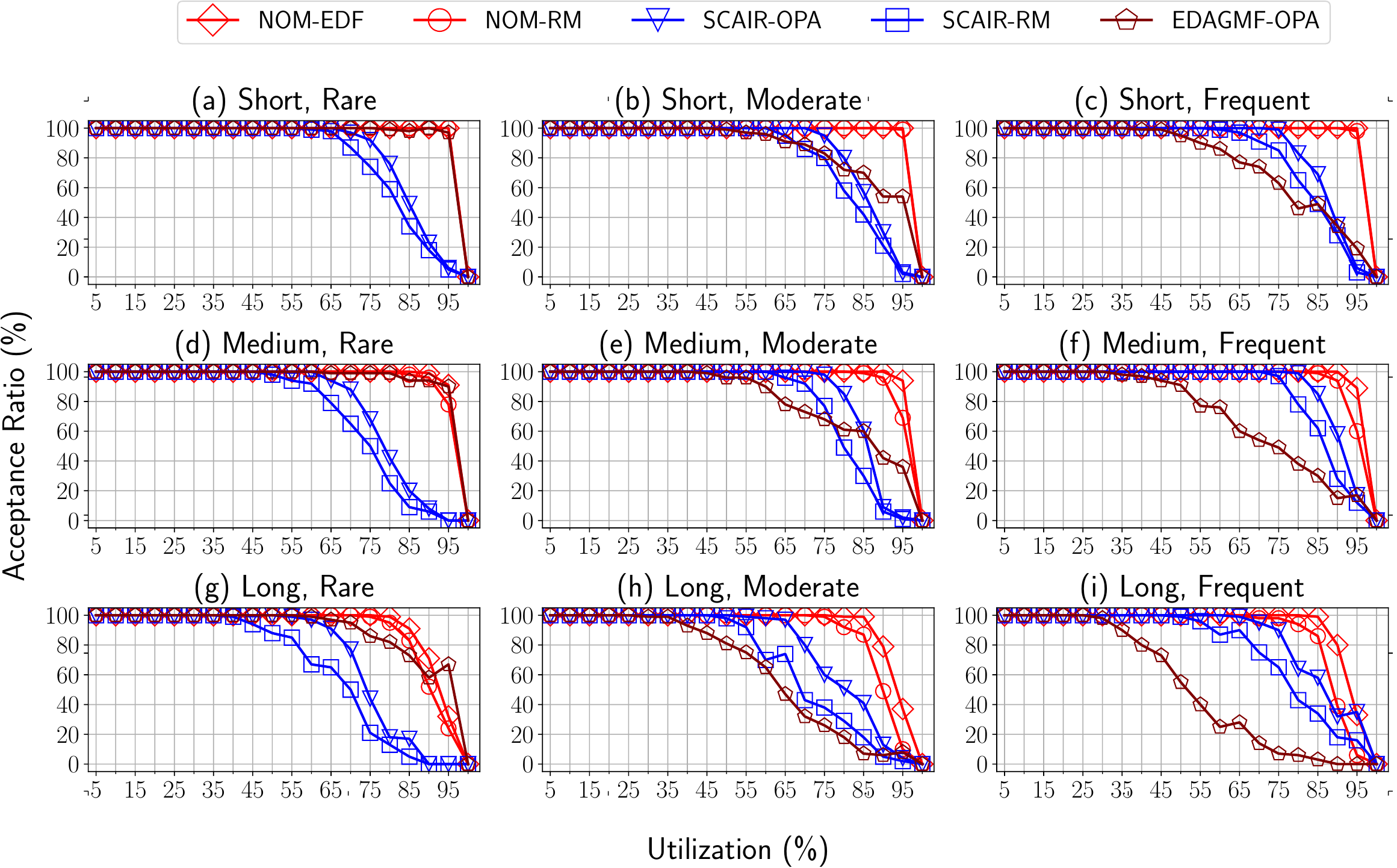}
	\caption{Acceptance ratio of the approaches under different task set configurations.}
	\label{fig:sched_results_full}	
\end{figure}

Figure~\ref{fig:sched_results_full} illustrates the acceptance ratio of different approaches across all nine task set configurations, excluding scenarios involving release jitter.

\end{document}